\documentclass[journal,onecolumn,draftclsnofoot]{IEEEtran}
\usepackage{etex}

\usepackage{mathtools}
\usepackage[thinc]{esdiff}
\usepackage{commath,amsmath,amssymb,amsfonts}
\usepackage{mathtools}
\usepackage{amsthm}
\usepackage{algorithmic}
\usepackage{pgfplots} 
\pgfplotsset{compat=1.15}
\usepackage{graphicx}
\usepackage{pgfgantt}
\usepackage{pdflscape}
\usepackage{pst-plot}
\usepackage{xfrac}
\usepackage{colortbl}
\usepackage{cancel}
\usepgfplotslibrary{fillbetween}
\usepackage{amssymb,bm}
\usepackage{float}
\usepackage{amsmath}
\usepackage{hyperref}
\usepackage{multirow}
\usepackage{xcolor}
\usepackage{tikz}
\usepackage{mathrsfs}
\usepackage{bbm}

\usepackage{cite} 

\usepackage{comment}

\def \B{\mathcal{B}}

\def \D{\mathcal{D}}
\def \E{\mathcal{E}}

\def \M{\mathcal{M}}

\def \Q{\mathcal{Q}}

\def \S{\mathcal{S}}

\def \U{\mathcal{U}}

\def \W{\mathcal{W}}
\def \X{\mathcal{X}}
\def \Y{\mathcal{Y}}
\def \Z{\mathcal{Z}}

\def \fu{\mathbf{u}}
\def \fv{\mathbf{v}}
\def \fx{\mathbf{x}}
\def \fy{\mathbf{y}}

\def \fY{\mathbf{Y}}

\def \f0{\mathbf{0}}

\definecolor{blau_1a}{RGB}{93,133,195}
\definecolor{blau_2a}{RGB}{0,156,218}
\definecolor{gruen_3a}{RGB}{80,182,149}
\definecolor{gruen_4a}{RGB}{175,204,80}
\definecolor{gruen_5a}{RGB}{221,223,72}
\definecolor{orange_6a}{RGB}{255,224,92}
\definecolor{orange_7a}{RGB}{248,186,60}
\definecolor{rot_8a}{RGB}{238,122,52}
\definecolor{rot_9a}{RGB}{233,80,62}
\definecolor{lila_10a}{RGB}{201,48,142}
\definecolor{lila_11a}{RGB}{128,69,151}

\definecolor{blau_1b}{RGB}{0,90,169}
\definecolor{blau_2b}{RGB}{0,131,204}
\definecolor{gruen_3b}{RGB}{0,157,129}
\definecolor{gruen_4b}{RGB}{153,192,0}
\definecolor{gruen_5b}{RGB}{201,212,0}
\definecolor{orange_6b}{RGB}{253,202,0}
\definecolor{orange_7b}{RGB}{245,163,0}
\definecolor{rot_8b}{RGB}{236,101,0}
\definecolor{rot_9b}{RGB}{230,0,26}
\definecolor{lila_10b}{RGB}{166,0,132}
\definecolor{lila_11b}{RGB}{114,16,133}

\definecolor{mycolor1}{RGB}{249,245,233}

\usepackage{accents}

\theoremstyle{remark}	\newtheorem{theorem}{Theorem}
\theoremstyle{remark}	\newtheorem{lemma}[theorem]{Lemma}
\theoremstyle{remark}	
\theoremstyle{remark}	
\theoremstyle{remark} \newtheorem{definition}{Definition}
\theoremstyle{remark} \newtheorem{remark}{Remark}
\theoremstyle{remark} 

\usepackage{amssymb}

\usepackage{autobreak}
\allowdisplaybreaks

\begin{document}
% ---
\title{Deterministic Identification Over Poisson Channels} %For Molecular Communication

\author{
	\vspace{0.1cm}
    \IEEEauthorblockN{Mohammad J. Salariseddigh\IEEEauthorrefmark{1}, Uzi Pereg\IEEEauthorrefmark{1}, 
    Holger Boche\IEEEauthorrefmark{2}, Christian Deppe\IEEEauthorrefmark{1}, and Robert Schober\IEEEauthorrefmark{3}} \\
	\vspace{0.25cm}
    \IEEEauthorblockA{\normalsize \IEEEauthorrefmark{1} Institute for Communication Engineering, Technical University of Munich \\
    \IEEEauthorrefmark{2} Chair of Theoretical Information Technology, Technical University of Munich, and \\
     Munich Center for Quantum Science and Technology (MCQST) \\
    \IEEEauthorrefmark{3} Institute for Digital Communications, Friedrich-Alexander-University Erlangen-N{\"u}rnberg \\
    {\tt Email}: {\{mohammad.j.salariseddigh,~uzi.pereg,~boche,~christian.deppe\}@tum.de,\\ robert.schober@fau.de} \\
}}
% ---
\maketitle
% ---
\begin{abstract}
Deterministic identification (DI) for the discrete-time Poisson channel, subject to an average and a peak power constraint, is considered. It is established that the code size scales as $2^{(n\log n)R}$, where $n$ and $R$ are the block length and coding rate, respectively. The authors have recently shown a similar property for Gaussian channels \cite{Salariseddigh_arXiv_ITW}. Lower and upper bounds on the DI capacity of the Poisson channel are developed in this scale. Those imply that the DI capacity is infinite in the exponential scale, regardless of the dark current, i.e., the channel noise parameter.
\end{abstract}
% ---
% ---
\IEEEpeerreviewmaketitle
% ---
\section{Introduction}
In Shannon's communication paradigm \cite{S48}, a sender, Alice, encodes her message in a manner that will allow the receiver, Bob, to recover the message. In other words, the receiver's task is to determine which message was sent. However, in the identification setting, the coding scheme is designed to accomplish a different goal \cite{AD89}. The decoder's only task is to determine whether a \emph{particular} message was sent or not, while the transmitter does not know which message the decoder is interested in. Applications of identification are often associated with event-triggered systems. Specifically, in molecular communication (MC), information is transmitted via bio-chemical signals or molecules \cite{NMWVS12,FYECG16}. Research on micro-scale molecular technology, such as intra-body networks, is still in its early stages and faces many challenges. Nonetheless, MC is a promising contender for future applications, such as the sixth generation of cellular communication (6G) \cite{6G+,6G_PST}, and nanomedical applications including cancer treatment \cite{Hobbs_ea98,Jain99,Wilhelm16} and targeted drug delivery \cite{Muller04,Nakano13}.

Ahlswede and Dueck \cite{AD89} employ a randomized identification coding scheme, in which the encoder has access to a large source of random bits, and the transmission is generated using the output of this source. It is well-known that such a resource does not increase the transmission capacity \cite{A78}. On the other hand, Ahlswede and Dueck \cite{AD89} have established that given local randomness at the encoder, reliable identification is accomplished with a double-exponential coding scale in the block length $n$, i.e., $\sim 2^{ 2^{nR}}$ \cite{AD89}. This behavior differs radically from the conventional transmission setting, where the code size grows of only exponentially in the block-length, i.e., $\sim{2^{nR}}$. Beyond the exponential gain in identification, the extension of the problem to various scenarios reveals a different property of the identification capacity, compared to the transmission capacity \cite{feedback,correlation,BV00,BD18_2,W04,BL17}. For instance, strictly causal feedback from the receiver to the transmitter \emph{can} increase the identification capacity of a memoryless channel \cite{feedback}, as opposed to the transmission capacity \cite{S56}. Nevertheless, it can be difficult to implement randomized-encoder identification codes with such a performance. The construction of identification codes is considered in \cite{SCR20-2,VK93,KT99,Bringer09,Bringer10}. Furthermore, identification for Gaussian channels is studied in \cite{MasterThesis,LDB20,Labidi2021,Ezzine2021,BV00,Salariseddigh_ITW,Salariseddigh_arXiv_ITW}.

In the deterministic coding setup of identification, for a discrete memoryless channel (DMC), the code size grows only exponentially in the blocklength, as in a conventional transmission \cite{AD89,AN99,Salariseddigh_arXiv_ICC,J85,Bur00}. However, the achievable identification rates are significantly higher compared to that of the transmission rates \cite{Salariseddigh_arXiv_ICC}. Deterministic codes often have the advantage of simpler implementation and simulation \cite{PP09}, explicit construction \cite{A09}, and a single-block reliable performance. In recent works by the authors \cite{Salariseddigh_ITW,Salariseddigh_arXiv_ITW,Salariseddigh_ICC,Salariseddigh_arXiv_ICC}, we have considered deterministic identification (DI) for channels with an average power constraint, including DMCs and Gaussian channels with fast or slow fading. In the Gaussian case, we have shown that in this scale, the code size scales as $2^{(n\log n)R}$, by deriving finite bounds on the DI capacity.

The Poisson channel is a useful model for an optical communication link with pulse amplitude modulation (PAM) and a direct-detection receiver \cite{Fillmore69,Gagliardi76,Shamai90,Shapiro09,Wyner88_I,Wyner88_II,Massey81,Verdu99}. In particular, the Poisson channel with pulse position modulation (PPM) is investigated in \cite{Bar73}, which is advanced in \cite{Pierce81}. The continuous-time Poisson channel (CTPC) with infinite bandwidth exposed to a peak-power constraint and a possible average power constraint on its input signal is extensively studied in the literature. Capacity of a general CTPC where the input signal is not restricted to be pulse amplitude modulation (PAM) is derived exactly. For instance, computation of capacity using different methods is addressed first in \cite{Kabanov78} where the capacity under a peak power constraint and PPM restriction using martingale techniques is obtained. Those results subject to both peak and average constraints are extended in \cite{Davis80}. Further, in \cite{Wyner88_I,Wyner88_II}, the capacity is derived using information bearing ``rate functions" of unbounded bandwidth and the entire reliability function for all rates below the capacity as well as the error exponent and construction of specific capacity-achieving codes are given. A similar treatment is conducted in \cite{Wyner84}. In \cite{Wyner88_I}, it is shown that a two level modulation on rate $\lambda(t)$ is capacity-achieving. To emphasize, the models which take into account the bandwidth limitations are investigated in \cite{Butman82,Lesh83,Lipes80,Shamai91,Shamai93}. The associated capacities, optimum coding and modulation methods for this type of channels are sharply different than that of the unbounded bandwidth cases.

The reliability function of an ideal Poisson channel with noiseless feedback is determined in \cite{Lapidoth93}. An extension to multi-user settings is given in \cite{Lapidoth98,Lapidoth03_4,Bross09}. Extensions to multi-user networks and channel side information are studied in \cite{Lapidoth98,Lapidoth03_4}. The Poisson channel with a random dark current, which is not known to the transmitter or the receiver, is considered in \cite{Frey91} and a closed-form capacity formula is derived when the the dark current is deterministic. The capacity under an $L_p$-norm constraint is studied in \cite{Frey92}. Bounds on the transmission capacity of a channel with an average power constraint were established in \cite{Brady90}. An improved lower bound on the transmission capacity, and determined an asymptotic upper bound for the discrete-time Poisson channel (DTPC) in the limit of infinite constraints, i.e., when average and peak power constraints both tend to infinity while their ratio remains constant is derived in \cite{Lapidoth03_1,Lapidoth03_2}. The Poisson multiple-access channel is studied in \cite{Lapidoth98}. The asymptotic capacity in the limit of a low-value average and possibly peak power constraint, or at both average and peak power constraints, is considered in \cite{Lapidoth08_1,Lapidoth11}. Optimal codes for a DTPC with memory under peak and average power constraints is constructed in \cite{Ahmadypour20}. Bounds on the capacity of the compound DTPC are given in \cite{Etemadi19}. The Poisson channel with fading is addressed in \cite{Chakraborty07}. Covert communication over Poisson channel is studied in \cite{Wang21}.

The Poisson channel is relevant for practical 6G networks in the context of MC \cite{Gohari16} and optical communications \cite{Cao13,Mceliece81}. The functionality of some 6G applications requires only a message identification. In particular, in the context of MC, a nano-device might seek only to formulate its knowledge about a specific task in terms of a reliable Yes/No answer. For such tasks, the identification problem and related code design plays a crucial role, and is deemed as a potential key technology for 6G. For instance, in the course of targeted drug delivery \cite{Muller04,Nakano13} or cancer treatment \cite{Hobbs_ea98,Jain99,Wilhelm16}, a nano-device seeks to know whether a specific drug is released or not, whether another nano-device has replicated itself or not, whether a certain molecule was detected or not, whether the molecular storage is empty or not, whether the blood PH exceeds a critical threshold or not \cite{Nakano14}, etc.

In this work, we consider DI over a DTPC under average and peak power constraints.
In particular, we establish that the code size scales as $2^{(n\log n)R}=n^{nR}$, by deriving finite positive bounds on the DI capacity of the DTPC. The approach is similar to our previous analysis for Gaussian channels \cite{Salariseddigh_ITW,Salariseddigh_arXiv_ITW}. However, the analysis and the upper bound in the present paper are different. Here, in the achievability proof, we consider packing of hyper spheres with radius $\sim n^{\frac{1}{4}}$ inside a larger hyper cube. While the radius of the small spheres in the Gaussian derivation in \cite{Salariseddigh_arXiv_ITW} tends to zero, the radius here grows in the block-length, $n$.
%We observe that the volume of a hyper sphere with radius $\sim n^c$ tends to zero for $0 < c < \frac{1}{2}$.
Yet, we can pack a super-exponential number of spheres within a larger cube.
In the converse part, the derivation for the DTPC is more involved and leads to a larger upper bound. Instead of establishing a minimum distance between the codewords, as in the Gaussian derivation \cite{Salariseddigh_arXiv_ITW}, we use the letter-wise ratio.

The rest of this paper is structured as follows. Section~2 introduces channel model and contains required definitions for coding. Section~3 studies the previous contributions and results for the Shannon capacity of CTPC/DTPC. The main results
for the DTPC are presented in Section~4. Finally, in Section~5 the paper concludes with summary and important remarks.
% ------------------------------------
\section{Definitions}
\label{sec:preliminaries}
In this section, we introduce the channel model and coding definitions.
%%%%%%%%%%%%%%%%%%%%%%%%%%%%%%%%%%%%%%%%%%%%%%%%%%%%%
\subsection{Molecular Transmission and Reception}
\label{Subsec.Poisson_Concentration_Transmitter}
In MC, the transmitter has a reservoir that consists of a limited number of particles. Hence, the transmitter controls the number of released particles via a biological outlet, the size of which is controlled by the actual values of the transmitted codewords. Given a sufficiently large number of particles in the reservoir and a small probability of each exiting through the storage, the Poisson distribution with a codeword-dependent mean is a standard model that represents the number of released particles stochastically \cite{Gohari16}.

It should be emphasized that the particle generator block is directly related to the bio-chemical structure of the nano-transmitter, and it is impossible to manipulate this structure to extract additional randomness. Therefore, the encoding procedure is entirely deterministic. The transmission is constrained by the maximal number of particles released in each time slot. 

%While a limitation with respect to the entire length of a codeword is represented through imposing the corresponding arithmetic mean of the codewords as given in the second condition.

The receiver absorbs any molecule that hits its surface, and counts the total number of molecules that have hit it so far. The fraction of hitting molecules can be derived by solving \emph{Fick's law} subject to initial and boundary conditions obeying the absorbing process following the techniques in \cite{Schulten00,Redner01}.
% -------------------
\subsection{Notation}
In this section, we introduce the channel models and coding definitions. We use the following notation conventions throughout. Calligraphic letters $\X,\Y,\Z,\ldots$ are used for finite sets. Lowercase letters $x,y,z,\ldots$ stand for constants and values of random variables, and uppercase letters $X,Y,Z,\ldots$ stand for random variables. The distribution of a random variable $X$ is specified by a probability mass function (pmf) $p_X(x)$ over a finite set $\X$. All logarithms and information quantities are to the base $2$. The set of consecutive natural numbers from $1$ through $M$ is denoted by $[\![M]\!]$. The set of whole numbers is denoted by $\mathbb{N}_{0} \triangleq \{0,1,2,\ldots\}$. The $\ell_2$-norm and $\ell_{\infty}$-norm are denoted by $\norm{\mathbf{x}}$ and $\norm{\mathbf{x}}_{\infty}$, respectively. Further, we denote the $n$-dimensional hypersphere of radius $r$ centered at $\fx_0$ with respect to the $\ell_2$-norm by
% ---
\begin{align}
    \S_{\fx_0}(n,r) = \{\fx\in\mathbb{R}^n : \norm{\fx-\fx_0} \leq r \} \,.\
\end{align}
% ---
An $n$-dimensional cube with edge $A$ is defined by
% ---
\begin{align}
    \Q_{\f0}(n,A) = \{\fx\in\mathbb{R}^n : 0 < x_t\leq A \} \,,\,
\end{align}
% ----
for all $t\in[\![n]\!]$. In the continuous case, we use the cumulative distribution function $F_X(x)=\Pr(X\leq x)$ for $x\in\mathbb{R}$, or alternatively, the probability density function (pdf) $f_X(x)$, when it exists.
% --------------------------
\subsection{Poisson Channel}
Here, we consider the Poisson channel $\W$ which arises as key channel model in the context of MC \cite{Gohari16}. We assume that the channel takes a positive real-valued input $X \in \mathbb{R}_{>0}$ and outputs a symbol $Y\in\mathbb{N}_0$ where $Y \sim \text{Pois}(\lambda+X)$. Here, $\lambda \in (0,\infty)$ is the dark current (see Figure~\ref{Fig.PoissonChannel}). The letter-wise channel law is given by
% ---
\begin{align}
    W(y|x)%\equiv\Pr(Y=y|X=x) 
    = \frac{e^{-(\lambda+x)}(\lambda+x)^y}{y!} \,.\,
\end{align}
% ---
Since the channel is memoryless, for $n$ channel uses, the law is given by 
% ---
\begin{align}
    \label{Eq.Poisson_Channel_Law}
    W^n(\fy|\fx) & = \prod_{t=1}^n W(y_t|x_t) 
    \nonumber\\
    & = \prod_{t=1}^n \frac{e^{-(\lambda+x_t)}(\lambda+x_t)^{y_t}}{y_t!} \;.\,
\end{align}
% ---
The peak and average power constraints on the codewords are
% ---
\begin{align}
    \label{Ineq.Constraints}
    0 < x_{t} \leq P_{\,\text{max}} \quad \text{and} \quad \frac{1}{n}\sum_{t=1}^{n} x_{t} \leq P_{\,\text{avg}} \,,\,
\end{align}
% ---
respectively, for all $t\in[\![n]\!]$, where $P_{\,\text{max}}, P_{\,\text{avg}} > 0$ represent the values for peak and average power constraints. We note that while the average power constraint for Gaussian channel is a non-linear (square) expression of symbols (signify the signal energy), here for the Poisson channel we have a linear cost function. This difference comes mainly from the type of transmitter for Poisson model when it sends either $X$ number of molecules or a realization of a Poisson random variable with mean value of $X$ \cite{Gohari16}.
% ---
\begin{figure}[H]
    \centering
	\scalebox{1}{
\tikzstyle{l} = [draw, -latex']
\tikzstyle{Block1} = [draw,top color=white, bottom color=white!80!orange, rectangle, rounded corners, minimum height=2em, minimum width=2.5em]
\tikzstyle{Block2} = [draw,top color=white, bottom color=cyan!30, rectangle, rounded corners, minimum height=2em, minimum width=2em]
\tikzstyle{Block3} = [draw,top color=white, bottom color=cyan!10, rectangle, rounded corners, minimum height=2em, minimum width=2em]
\tikzstyle{Block4} = [draw,top color=white, bottom color=red!20, rectangle, rounded corners, minimum height=2em, minimum width=2em]
\tikzstyle{Block5} = [draw,top color=white, bottom color=white!80!orange, rectangle, rounded corners, minimum height=2em, minimum width=2.5em]
\tikzstyle{input} = [coordinate]
\tikzstyle{sum} = [draw, circle,inner sep=0pt, minimum size=5mm,  thick]
\tikzstyle{arrow}=[draw,->]
\tikzstyle{small_node} = [draw, circle,inner sep=0pt, minimum size=.8mm,thick]
\begin{tikzpicture}[auto, node distance=2cm,>=latex']
\node[] (M) {$i$};
\node[Block1,right=.5cm of M] (enc) {Enc};
% \node[Block2, right=.7cm of enc] (Pois_Gen) {$\text{\small Diffusion CHL}$};
\node[Block2, right=.7cm of enc] (channel) {$\text{\small Diffusion Channel}$};
%\node[sum, right=.7cm of channel] (adder) {$+$};
%\node[Block4, right=.5cm of adder] (Pois) {\text{Pois(.)}};
\node[Block5, right=.7cm of channel] (dec) {Dec};
\node[below=.5cm of dec] (Target) {$j$};
\node[right=.5cm of dec] (Output) {\text{\small Yes/No}};
%\node[above=.5cm of adder] (noise) {$\lambda$};
%
\draw[->] (M) -- (enc);
\draw[->] (enc) -- node[above]{$\textbf{u}_i$} (channel);
%\draw[->] (Pois_Gen) -- (channel);
%\draw[->] (channel) -- node [above] {} (adder);
%\draw[->] (noise) -- (adder);
%\draw[->] (Pois_Gen) -- node [above] {\text{}} (channel);
%\draw[->] (adder) --node[above]{} (Pois);
\draw[->] (channel) --node[above]{$\fY$} (dec);
\draw[->] (dec) -- (Output);
\draw[->] (Target) -- (dec);
\end{tikzpicture}
}
	\caption{Deterministic identification for the Poisson channel. $Y(t)$ is Poisson distributed with rate $\lambda + u_i(t)$. Given the observation $\fY$, the decoder asks whether $j$ equals $i$ or not.}
	\label{Fig.PoissonChannel}
\end{figure}
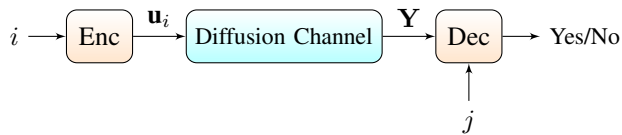
% ---
%%%%%%%%%%%%%%%%%%%%%%%%%%%%%%%%%%%%%%%%%%%%%%%%%%%%%%%%%%%
\subsection{Coding for the Poisson Channel}
The definition of a DI code for the Poisson channel is given below.
%%%
\begin{definition}[Poisson DI Code]
\label{GdeterministicIDCode}
An $\left( L(n,R),n\right)$ DI code for a Poisson channel $\W$ under average and peak power constraint of $P_{\,\text{ave}}$ and $P_{\,\text{max}}$, respectively, assuming $L(n,R)$ is an integer, is defined as a system $(\U,\mathscr{D})$ which consists of a codebook $\U=\{ \mathbf{u}_i \}_{i\in[\![L(n,R)]\!]}$, $\U\subset \X^n$, such that
% ---
\begin{align}
    \label{Ineq.Peak_Power_Const}
    0 < u_{i,t} \leq P_{\,\text{max}} \,,\, \\
    \label{Ineq.Ave_Power_Const}
    \frac{1}{n} \sum_{t=1}^{n} u_{i,t} \leq P_{\,\text{ave}} \,,\,
\end{align}
% ---
for all $i\in[\![L(n,R)]\!]$, all $t\in[\![n]\!]$, and a collection of decoding regions
$\mathscr{D}=\{ \D_i \}_{i\in[\![L(n,R)]\!]}$ with
% ---
\begin{align}
    \bigcup_{i=1}^{L(n,R)}\D_i\subset\mathbb{N}_0^n \;.\,
\end{align}
% ---
Given a message $i\in [\![L(n,R)]\!]$, the encoder transmits $\mathbf{u}_i$. The decoder's aim is to answer the following question: Was a desired message $j$ sent or not? There are two types of errors that may occur:
Rejection of the true message or accepting a false message. Those are referred to as type I and type II errors, respectively.

The error probabilities of the identification code $(\U,\mathscr{D})$ are given by
%%%
\begin{align}
    P_{e,1}(i) & = 1 -\sum_{\fy \in \D_i} W^n \left( \fy \, \big| \, \fu_i \right) &&\hspace{-2cm}  \text{(missed-identification error)} \,,\, \label{Eq.GTypeIErrorDef}
     %W^n(\setd_i^c|u_i)  &&\hspace{-2cm}  \text{(missed-identification error)} \,,\, \label{Eq.GTypeIErrorDef} 
    \\
    P_{e,2}(i,j) & = \sum_{\fy \in \D_j} W^n \left( \fy \, \big| \, \fu_i \right) && \hspace{-2cm} \text{(false identification error)} \,.\,
    %W^n(\setd_j|u_i)  &&\hspace{-2cm} \text{(false identification error)} \,.\,
    \label{Eq.GTypeIIErrorDef}
\end{align}
%%%
An $(L(n,R),n,\lambda_1,\lambda_2)$ DI code further satisfies
%%%
\begin{align}
\label{Eq.GTypeIError}
 P_{e,1}(i) &\leq \lambda_1 \,,\  \\
\label{Eq.GTypeIIError}
 P_{e,2}(i,j) &\leq \lambda_2 \,,\
\end{align}
%%%
for all $i,j\in [\![L(n,R)]\!]$, such that
$i\neq j$.
%%%

A rate $R>0$ is called achievable if for every $\lambda_1,\lambda_2>0$ and sufficiently large $n$, there exists an $(L(n,R),n,\lambda_1,\lambda_2)$ DI code. The operational DI capacity of the Poisson channel is defined as the supremum of achievable rates, and will be denoted by $\mathbb{C}_{DI}(\W,L)$. 
\end{definition}
\section{Previous Contributions}
In this section, we review known results for the capacity of a DTPC in the asymptotic regimes of $P_{\,\text{avg}} \to 0$ and $P_{\,\text{avg}} \to \infty$. The asymptotic capacity with an average-power constraint $P_{\,\text{avg}}$, when $P_{\,\text{avg}} \to\infty$ and $\frac{P_{\,\text{avg}}}{\lambda}$ is fixed was studied in \cite{Brady90}. Furthermore, the same problem for a constant $\lambda$, with and without an additional peak power constraint was studied in \cite{Lapidoth03_1}. The first-order asymptotic capacity for $P_{\,\text{avg}} \to 0$, both when $\frac{P_{\,\text{avg}}}{\lambda}$ is kept constant and when $\lambda$ is fixed, with and without a peak power constraint, was determined in \cite{Lapidoth11}. Later, the previously obtained first-order capacity approximation when $P_{\,\text{avg}} = \lambda$ is constant were improved in \cite{Wang14} where a refined approximation was determined, including an exact characterization of the second-order term, as well as an asymptotic characterization of the third-order term with respect to the dark current. Asymptotic upper bounds for the DTPC capacity with an average-power constraint were given in \cite{Lapidoth03_1,Lapidoth11,Wang14,Aminian15}. Explicit asymptotic and non-asymptotic capacity lower bounds for several settings were given in \cite{Lapidoth03_1,Lapidoth11,Wang14,Martinez07,Cao10,Yu14}. 

Previously, the best known non-asymptotic upper bound, which was in fact the best capacity upper bound outside the limiting case $P_{\,\text{ave}} \to 0$, was derived in \cite{Martinez07}. However, the proof suffers a small gap, as mentioned in \cite{Lapidoth03_1}, and is not considered completely rigorous. Later, in \cite[Th.~8]{Cheraghchi19} strictly tighter upper bounds than the bound in \cite{Martinez07} for all $P_{\,\text{avg}} > 0$ which are considered the best current capacity upper bounds for the DTPC with zero dark current term, i.e., $\lambda = 0$ subject to an average power constraint $P_{\,\text{avg}}$ for all values of $P_{\,\text{avg}}$ outside the limiting case $P_{\,\text{avg}} \to 0$. In the same paper, the result of \cite{Martinez07} was recovered as an special (sub-optimal) case, thus yielding a rigorous proof for the bound proposed in \cite{Martinez07}. As well, the same authors in \cite[Th.~1]{Cheraghchi20} showed derived a significantly improved non-asymptotic upper bound on the capacity of the DTPC with constant positive dark current $\lambda \geq 0$ and an average power constraint $P_{\,\text{avg}}$ in non-asymptotic regimes of $P_{\,\text{avg}}$.
% -------------------------------------------
\subsection{Capacity-Achieving Distributions}
There has been an enormous literature focusing on the properties of capacity-achieving distributions for different channels. This problem is well-understood for quite general classes of additive noise channels under several input constraints (see, e.g., the early works \cite{Shamai90,Smith71,Abou01} and the recent works \cite{Elmoslimany17,Fahs17,Dytso18}). For the DTPC, in the absence of input constraints, capacity is infinite. The DTPC under a peak power constraint alone was addressed in \cite{Dytso21}, and was shown that the support size is of an order between $\sqrt{P_{\,\text{max}}}$ and $\sim P_{\,\text{max}} \log^2 P_{\,\text{max}}$. In particular, they characterized the capacity in terms of the output optimal distribution where capacity equals $-\log P_{Y^*}(0)$ for $P_{Y^*}(0)$ to be the optimal output distribution. An analytic expression for the transmission capacity of a the DTPC with an average power constraint alone, is still open. However, several bounds and asymptotic behaviors for the DTPC in different setups have been established. For instance, it was shown that a capacity-achieving input distribution for the DTPC under an average power constraint must have a finite support. The number of mass points depends on the average and the peak power constraints, and increases to infinity as the constraints are relaxed \cite{Shamai90}. It was conjectured by Shamai \cite{Shamai90} that the support of such a distribution must be countably infinite. The result was extended in \cite{Cao13_1,Cao13_2} and it was stated that the support of such a distribution for the DTPC under an average-power constraint must have an \emph{unbounded support}. Moreover, they also proved that such a distribution has a non-zero mass at $x = 0$, and, if a peak power constraint $P_{\,\text{max}}$ is present, at $x = P_{\,\text{max}}$ as well. Unlike additive noise channels, less is known about the capacity-achieving distributions of a DTPC when there is only an average-power constraint. In \cite{Cheraghchi19} it was shown that such a distribution for the DTPC with an arbitrary $\lambda \geq 0$, under an average-power constraint and/or a peak power constraint, is discrete (see the conjecture by Shamai \cite{Shamai90}). It was further shown that the support of such a capacity-achieving distribution for the DTPC under an average-power constraint and/or a peak power constraint has a finite intersection with every bounded interval \cite[see Th.~14]{Cheraghchi19}. Further discussions on the capacity-achieving distributions are provided in \cite{Cao13}. Here, we will consider the identification setting, where the receiver is not required to determined the message, but rather identifies a specific task.
% -------------------------------------
\subsection{Asymptotic Characterizations}
In the sequel, we denote the capacity under an average power constraint $P_{\,\text{ave}}$, a peak power constraint $P_{\,\text{max}}$ and the dark current $\lambda$ by $\mathbb{C}\left( \lambda, P_{\,\text{ave}}, P_{\,\text{max}}\right)$.
For the small values of $P_{\,\text{ave}}$, in \cite{Lapidoth11} it was shown that the asymptotic capacity of the DTPC, i.e., when the average input power tends to zero while the peak-power, if finite, is fixed, scales as $- P_{\,\text{ave}} \log P_{\,\text{ave}}$, i.e.,
% ---
\begin{align}
    \lim_{P_{\,\text{ave}} \to 0} \frac{\mathbb{C}(\lambda = cP_{\,\text{ave}},P_{\,\text{ave}},P_{\,\text{max}})}{P_{\,\text{ave}}\log P_{\,\text{ave}}} = -1 \;,\,
\end{align}
% ---
for any $c\in[0,\infty)$ and $P_{\,\text{max}} \in (0,\infty]$. Furthermore, they provided the following upper bound
% ---
\begin{align}
    \label{Ineq.UB_Lapidoth}
    \mathbb{C}(0,P_{\,\text{ave}},\infty) \leq -P_{\,\text{ave}} \log p - \log(1-p) + \frac{P_{\,\text{ave}}}{\beta} + P_{\,\text{ave}} \cdot \max \left(0,\left( \frac{1}{2} \log \beta + \log \left( \frac{\overline{\Gamma}\left( \frac{1}{2}, 1/\beta \right)}{\sqrt{\pi}} + \frac{1}{2\beta} \right) \right) \right) \,,\,
\end{align}
% ---
where $p \in (0,1)$ and $\beta > 0$ are arbitrary constants and $\overline{\Gamma}(.)$ is the upper incomplete Gamma function. Later, in \cite{Wang14}, for a small value of $P_{\,\text{ave}}$ the higher order asymptotic behavior for $\mathbb{C}(P_{\,\text{ave}})$ was characterized and given by
% ---
\begin{align}
    \label{Eq.Asymptotic_Wang}
    \mathbb{C}(\lambda = cP_{\,\text{ave}},P_{\,\text{ave}},P_{\,\text{max}}) = -P_{\,\text{ave}} \log P_{\,\text{ave}} - P_{\,\text{ave}} \log \left( - \log P_{\,\text{ave}} \right) + \mathcal{O}(P_{\,\text{ave}}) \,,\,
\end{align}
% ---
where $c \in [0,\infty)$. This upper bound holds irrespective of whether a peak power constraint is imposed or not as long as $P_{\,\text{max}}$ is positive and does not approach zero together with $\epsilon$. Also the following upper bound was stated.
% ---
\begin{align*}
    \mathbb{C}(\lambda = cP_{\,\text{ave}},P_{\,\text{ave}},P_{\,\text{max}}) \hspace{-.3mm} \leq \hspace{-.3mm} P_{\,\text{ave}} \hspace{-.3mm} - \hspace{-.3mm} P_{\,\text{ave}} \log \log P_{\,\text{ave}} - \log (1 - P_{\,\text{ave}}) \hspace{-.3mm} - \hspace{-.3mm} P_{\,\text{ave}} \log \left( 1 - \frac{1}{\log P_{\,\text{ave}}} \right) \hspace{-.3mm} + \hspace{-.3mm} P_{\,\text{ave}} \cdot \sup_{x \geq 0} \phi_{\mu}(x) \,,\,
\end{align*}
% ---
which matches the asymptotic behavior given in (\ref{Eq.Asymptotic_Wang}) where
\begin{align}
    \phi_{\mu}(x) \coloneqq \frac{1-e^{-x}}{x} \log \left(\frac{-x}{P_{\,\text{ave}} \log P_{\,\text{ave}}} \right) \;.\,
\end{align}
% ---

For a large value of $P_{\,\text{ave}}$, and in the absence of a peak power constraint, it was shown in \cite{Lapidoth08_2} that
% ---
\begin{align}
    \lim_{P_{\,\text{ave}}\to\infty} \left[ \mathbb{C}(\lambda, P_{\,\text{ave}},\infty) - \frac{1}{2} \log P_{\,\text{ave}} \right] = 0 \,,\,
\end{align}
% ---
where the dark current is a non-negative constant, i.e., $\lambda \geq 0$. Expressions for the capacity, based on the ratio of the average and the peak power constraint, i.e., $\alpha = \frac{P_{\,\text{ave}}}{P_{\,\text{max}}}$ was provided in \cite{Lapidoth03_2} as follows
% ---
\begin{align}
    \mathbb{C}(\lambda,P_{\,\text{ave}},P_{\,\text{max}}) =
    \begin{cases}
        \frac{1}{2} \log P_{\,\text{max}} + (\alpha-1)u  - \log \left( \frac{1}{2} - \alpha u \right) - \frac{1}{2} \log \left( 2\pi e \right) + \mathcal{O}(1) \qquad & \qquad 0 < \alpha < \frac{1}{3} \,,\ \\
        \frac{1}{2} \log P_{\,\text{max}} - \frac{1}{2} \log \frac{\pi e}{2} + \mathcal{O}(1) \qquad & \qquad \frac{1}{3} \leq \alpha \leq 1 \,,\,
    \end{cases}
\end{align}
% ---
for $u$ is the non-zero solution to
% ---
\begin{align}
    \sqrt{\pi}\, \text{erf}\left(\sqrt{u}\right) \left( \frac{1}{2} - \alpha u \right) - \sqrt{u} e^{-u} = 0 \,,\,
\end{align}
% ---
with $\text{erf}(.)$ being the Gauss error function.
Observe that $\alpha \ll 1$ represents the regimes of very weak peak power constraints, whereas $\alpha = 1$ corresponds to the absence of an average power constraint. Also, the term $\mathcal{O}(1)$ vanishes as $P_{\,\text{ave}},P_{\,\text{max}} \to \infty$ where $\alpha$ is considered to be constant.

The previously best upper bound in the presence of only average power constraint $P_{\,\text{ave}}$ in any regime outside the regime $P_{\,\text{ave}} \to 0$ was derived \cite[see Eq.~(10)]{Martinez07} and is given by
% ---
\begin{align}
    \label{Eq.Martinez}
    \mathbb{C}\left(\lambda, P_{\,\text{ave}}, \infty \right) \leq \left( P_{\,\text{ave}} + \frac{1}{2} \right) \log \left( P_{\,\text{ave}} + \frac{1}{2} \right) - P_{\,\text{ave}} \log P_{\,\text{ave}} - \frac{1}{2} + \log \left( 1 + \frac{\sqrt{2e}-1}{\sqrt{1+2P_{\,\text{ave}}}} \right) \,,\,
\end{align}
% ---
which is strictly less than the upper bound given in (\ref{Ineq.UB_Lapidoth}) for all $P_{\,\text{ave}} > 0$ and tends to $\frac{1}{2} \log (1 + P_{\,\text{ave}})$ for $P_{\,\text{ave}} \to \infty$. However, recently, authors in \cite{Cheraghchi19} yielded the best known upper bound on the capacity $\mathbb{C}\left(0, P_{\,\text{ave}}, \infty \right)$ in the absence of dark current, i.e., $\lambda = 0$ for any $P_{\,\text{ave}}$ outside the asymptotic regime $P_{\,\text{ave}} \to 0$, as follows
% ---
\begin{align}
    \mathbb{C} \left( 0, P_{\,\text{ave}}, \infty \right) \leq P_{\,\text{ave}} \ln \hspace{-.3mm} \left( \frac{1 + \left(1 + e^{1+\gamma} \right) P_{\,\text{ave}}+ 2 P_{\,\text{ave}}^2 }{e^{1+\gamma}P_{\,\text{ave}} \hspace{-.2mm} + \hspace{-.2mm} 2P_{\,\text{ave}}^2}\right) \hspace{-.5mm} + \ln \hspace{-.5mm} \left( \hspace{-.3mm} 1 \hspace{-.3mm} + \hspace{-.3mm} \frac{1}{\sqrt{2e}} \hspace{-.3mm} \left( \sqrt{\frac{1 + \left(1 + e^{1+\gamma} \right) P_{\,\text{ave}}+ 2 P_{\,\text{ave}}^2}{1 + P_{\,\text{ave}}}} - 1 \hspace{-.3mm} \right) \hspace{-.2mm} \right)
\end{align}
% ---
where $\gamma \approx 0.5772$ is the Euler-Mascheroni constant. This bound outperforms the previously best upper bound given in (\ref{Eq.Martinez}) and recovers that result with a rigorous proof. Further, in \cite[Th.~1]{Cheraghchi20}, the best upper bound on the capacity $\mathbb{C}\left(\lambda, P_{\,\text{ave}}, \infty \right)$ for a positive dark current term, i.e., $\lambda \geq 0$ is given by
% ---
\begin{align}
    \mathbb{C}\left(\lambda, P_{\,\text{ave}}, \infty \right) \leq \ln\left( \delta_{\lambda} + \frac{1}{\sqrt{2e}} \left( \frac{1}{\sqrt{1 - q_{\lambda,P_{\,\text{ave}}}}} - 1 \right) \right) - \left( P_{\,\text{ave}} + \lambda \right)\ln q_{\lambda,P_{\,\text{ave}}}
\end{align}
% ---
where $\delta_{\lambda} = e^{-\lambda e^{\lambda} E_1\left(\lambda\right)}$, with $E_1(z)$ being the exponential integral function and $q_{\lambda,P_{\,\text{ave}}}$ given by
% ---
\begin{align}
    q_{\lambda,P_{\,\text{ave}}} = 1 - \frac{1}{1 + e^{1+\gamma}\left( P_{\,\text{ave}} + \lambda \right) + \frac{2 - e^{1+\gamma}}{1 + P_{\,\text{ave}} + \lambda} \left(P_{\,\text{ave}} + \lambda\right)^2}
\end{align}
% ---
where $\gamma$ is the Euler-Mascheroni constant.

Poisson channel under an average power constraint alone is considered in \cite{Brady90,Brady90_PhD_Diss} where it was assumed that both $P_{\,\text{ave}}$ and $\lambda$ tend to infinity while their ratio, $\frac{P_{\,\text{ave}}}{\lambda}$, defined as SNR, is kept constant. Namely, for every $\epsilon > 0$ capacity is lower-bounded by
% \cite[see Th.~4 in Sec.~4]{Brady90_PhD_Diss}
% ---
% \begin{align}
%     \mathbb{C}\left(\lambda, P_{\,\text{ave}}, \infty \right) \leq \ln\left(1 + P_{\,\text{ave}} + \lambda \right) + \left( P_{\,\text{ave}} + \lambda \right) \ln\left( 1 + \frac{1}{P_{\,\text{ave}} + \lambda} \right) - \frac{1}{2} \ln\left( 2\pi \left( P_{\,\text{ave}} + \lambda \right)\right) + \ln\left(\frac{3}{2} \right) + \epsilon
% \end{align}
% ---
% holds for all $P_{\,\text{ave}} > C_{\epsilon}$, where $C_{\epsilon}$ is a large constant which depends on $\epsilon$.
%
% Bounds on the capacity was established as follows: Given $\epsilon > 0$ there exist an $P_{\epsilon}$ such that
% ---
\begin{align}
     \mathbb{C}\left(\lambda, P_{\,\text{ave}}, \infty \right) \geq \frac{1}{2} \log \frac{P_{\,\text{ave}}}{2\pi} - \frac{1}{2}\log\left( 1 + \frac{1}{\text{SNR}} \right) - \epsilon  \,,\,
\end{align}
% ---
and upper-bounded by
% ---
\begin{align}
     % \frac{1}{2} \log \frac{P}{2\pi} - \frac{1}{2}\log\left( 1 + \frac{1}{\text{SNR}} \right) - \epsilon \leq
     \mathbb{C}\left(\lambda, P_{\,\text{ave}}, \infty \right) \leq \frac{1}{2} \log \frac{P_{\,\text{ave}}}{2\pi} + \log \left( \sqrt{\text{SNR}} \left( 1 + \frac{1}{P_{\epsilon}} \right) + \frac{1}{\sqrt{\text{SNR}}} \right) + 1 + \log \frac{3}{2} + \epsilon \;,\,
\end{align}
% ---
where the $P_{\epsilon} < P_{\,\text{ave}}$ is a large constant depending on $\epsilon$ but strictly less that the average power constraint $P_{\,\text{ave}}$. 

An upper bound on the capacity of a point-to-point DTPC $W(y|x)$ under an average power constraint $P_{\,\text{ave}}$ and dark current $\lambda$ was proposed in \cite[Example~2]{Aminian15} where it was shown that the Shannon capacity satisfies
% ---
\begin{align}
    \mathbb{C}\left( \lambda, P_{\,\text{ave}}, \infty \right) & = \max_{X\;:\,\mathbb{E}[X] \leq P_{\,\text{ave}}} I(X,Y)
    \nonumber\\
    & \leq \max_{X\;:\,\mathbb{E}[X] \leq P_{\,\text{ave}}} \text{Cov} \left( X+\lambda \,,\, \log \left( X+\lambda \right) \right) \,,\,
\end{align}
% ---
with $Y \sim \text{Pois} (\lambda + X)$ and $\text{Cov}(X,Y) = \mathbb{E}\left[XY\right] - \mathbb{E}\left[X\right]\mathbb{E}\left[Y\right]$.
Further, for a DTPC with an average input constraint $P_{\,\text{ave}}$ and peak power constraint $P_{\,\text{max}}$, the following upper bound was reported
% ---
\begin{align}
    \mathbb{C}\left( \lambda, P_{\,\text{ave}}, P_{\,\text{max}} \right) & = \max_{\substack{ X\;:\, \mathbb{E}[X] \leq P_{\,\text{ave}} \,,\ \\ 0 \leq X \leq P_{\,\text{max}}}} I(X,Y)
    \nonumber\\
    & \leq
    \begin{cases} \frac{P_{\,\text{ave}}}{P_{\,\text{max}}}(P_{\,\text{max}} - P_{\,\text{ave}})\log(\frac{P_{\,\text{max}}}{\lambda}+1) \qquad & \qquad P_{\,\text{ave}} \leq \frac{P_{\,\text{max}}}{2} \;,\, \\ \frac{P_{\,\text{max}}}{4} \log(\frac{P_{\,\text{max}}}{\lambda}+1) \qquad & \qquad P_{\,\text{ave}} \geq \frac{P_{\,\text{max}}}{2} \;.\, 
    \end{cases}
\end{align}
% ---
The capacity of the direct-detection Poisson photon-counting channel with fading (Poisson fading channel) is established in \cite{Chakraborty07} where a single-letter characterization of the capacity, assuming a perfect CSI at the receiver is provided with perfect and no CSI at the transmitter. Also, the limiting behavior of the capacity in the high and low peak-signal-to-dark-noise ratio (SNR) regimes, namely in the limits as $\lambda \to 0$ and $\lambda \to \infty$ is addressed. The capacity for perfect CSI at the transmitter is reported to be
% ---
\begin{align}
    \mathbb{C}\left( \lambda, P_{\,\text{ave}}, P_{\,\text{max}} \right) = \max_{\substack{\mu\;:\, \mathbb{R}_0^* \to [0,1] \,,\ \\ \mathbb{E}\left[ \mu(S) \leq \sigma \right]}} \mathbb{E} \left[ \mu(S) \zeta\left(S\alpha, \lambda \right) - \zeta \left(\mu(S) S \alpha, \lambda \right) \right] \,,\,
\end{align}
% ---
where $\zeta(x,y) \coloneqq (x+y) \ln (x+y) - y\ln y$ for $x,y > 0$ with $0\ln 0 \coloneqq 0$. The $0 \leq \sigma \leq 1$ is the ratio of average to peak power constraint, and $S$ is a distribution satisfying following conditions
% ---
\begin{align}
    \Pr \left[ S > 0 \right] = 1 \,,\, \\
    \mathbb{E} \left[ S \right] < \infty \,,\, \\
    \mathbb{E} \left[ |\zeta(S\alpha, \lambda)| \right] < \infty \,.\,
\end{align}
% ---
Finally, the capacity for no CSI at the transmitter is given by
% ---
\begin{align}
    \mathbb{C}\left( \lambda, P_{\,\text{ave}}, P_{\,\text{max}} \right) = \max_{0 \leq \mu \leq \sigma} \mathbb{E} \left[ \mu \zeta\left(S\alpha, \lambda \right) - \zeta \left(\mu S \alpha, \lambda \right) \right] \,.\,
\end{align}
% ---
% -------------------------------------
\section{Main Result}
%%%
We develop lower and upper bounds on the achievable identification rates for the Poisson channel. Our DI capacity theorem is stated below.
%%%
\begin{theorem}
\label{Th.PDICapacity}
The DI capacity of the DTPC $\W$ subject to an average power constraint of $\frac{1}{n} \sum_{t=1}^n u_{i,t} \leq P_{\,\text{ave}}$ and a peak power constraint of $0 < u_{i,t} \leq P_{\,\text{max}}$ in the super exponential scale, i.e., $L(n,R)=2^{(n\log n)R}$, is bounded by
% ---
\begin{align}
    \frac{1}{4} \leq \mathbb{C}_{DI}(\W,L) \leq \frac{3}{2} \,.\;
\end{align}
% ---
% Hence, the DI capacity is infinite in the exponential scale  and zero in the double exponential, i.e.,
% %%%
% \begin{align}
%     \label{Eq.PDICapacity2}
%     \mathbb{C}_{DI}(\W,L) = 
%     \begin{cases}
%     \infty & \text{ for 
% 	$L(n,R)=2^{(n\log n)R}$.}\\
% 	0&\text{ for 
% 	$L(n,R)=2^{2^{(n\log n)R}}$.}
% \end{cases}
% \end{align}
% %%%
\end{theorem}
%%%
Hence, the DI capacity is infinite in the exponential scale  and zero in the double exponential scale, i.e.,
% ---
\begin{align}
    \label{Eq.GDICapacityFast2}
    \mathbb{C}_{DI}(\W,L) = 
    \begin{cases}
    \infty & \text{ for 
	$L(n,R)=2^{nR}$} \,,\, \\
	0&\text{ for 
	$L(n,R)=2^{2^{nR}}$} \,.\,
\end{cases}
\end{align}
% ---
The proof of Theorem~\ref{Th.PDICapacity} is given below. The second part of the theorem is a direct consequence of the arguments in \cite[Lem.~3]{Salariseddigh_arXiv_ITW}.
% ---
\begin{proof}[Achievability Proof] 
Consider the Poisson channel $\W$. We show achievability using a packing of hyper spheres and a distance-decoder. We pack hyper spheres with radius $\sim n^{\frac{1}{4}}$ inside a larger hyper cube. While the radius of the small spheres in our previous derivation for the Gaussian channels vanishes \cite{Salariseddigh_arXiv_ITW}, the radius here diverges to infinity with $n$. Yet, we can obtain a positive rate while packing a super-exponential numbers of small spheres. A DI code for the Poisson channel $\W$ is constructed as follows.
%%%%%%
\subsubsection*{Codebook construction}
\label{Subsec.CodebookConstruction_1}
Observe that if $P_{\,\text{ave}} \geq P_{\,\text{max}}$, then the hypercube is fully inscribed within the hypersphere, hence, the setup in the presence of both average and peak power constraints as given in (\ref{Ineq.Peak_Power_Const}) and (\ref{Ineq.Ave_Power_Const}) is reduced to only the peak power constraint. For the case $P_{\,\text{ave}} < P_{\,\text{max}}$, we define
% ---
\begin{align}
    \label{Eq.A}
    A = \min \left(P_{\,\text{ave}},P_{\,\text{max}} \right) \,,\,
\end{align}
% ---
since by $0 < x_t \leq A$ for all $t \in [\![n]\!]$, both power constraints are met, namely $\frac{1}{n} \sum x_t \leq P_{\,\text{ave}}$ and $0 < x_t \leq P_{\,\text{max}}$ for all $t \in [\![n]\!]$. Hence, we restrict ourselves to a hypercube with edge length $A$. We use a packing arrangement of non-overlapping hyperspheres of radius $r_0 = \sqrt{n\epsilon_n}$ in a hypercube with edge length $A$, where
\begin{align}
    \epsilon_n = \frac{A}{n^{\frac{1}{2}(1-b)}} \,,\,
\end{align}
and $b>0$ is arbitrarily small.

Let $\mathscr{S}$ denotes a sphere packing, i.e., an arrangement of $L$ non-overlapping spheres $\S_{\fu_i}(n,r_0)$, for $i\in [\![L(n,R)]\!]$, with radius $r_0 = \sqrt{n\epsilon_n} = n^{\frac{1}{4}(1+b)} \sqrt{A}$, that cover a larger cube $\Q_{\f0}(n,A)$ with an edge length $A$ (see Figure~\ref{Fig.Density}).

As opposed to standard sphere packing coding techniques, the small spheres are not necessarily entirely contained within the larger cube. That is, we only require that the centers of spheres are inside $\Q_{\f0}(n,A)$ and are disjoint from each other and have a non-empty intersection with volume $\text{Vol}\left[ \Q_{\f0}(n,A))\right]$.

The packing density $\Delta_n(\mathscr{S})$ is defined as the fraction of the larger cube volume $\text{Vol}\left[\Q_{\f0}(n,A)\right]$ that is covered by the small spheres (see \cite[Ch.~1]{CHSN13}), i.e.,
%%%
\begin{align}
    \Delta_n(\mathscr{S}) \triangleq \frac{\text{Vol}\left(\Q_{\f0}(n,A)\cap\bigcup_{i=1}^{L}\S_{\fu_i}(n,r_0)\right)}{\text{Vol}\left[\Q_{\f0}(n,A)\right]} \,.\,
    \label{Eq.DensitySphereFast}
\end{align}
%%%
A sphere packing is called \emph{saturated} if no spheres can be added to the arrangement without overlap.
% ---
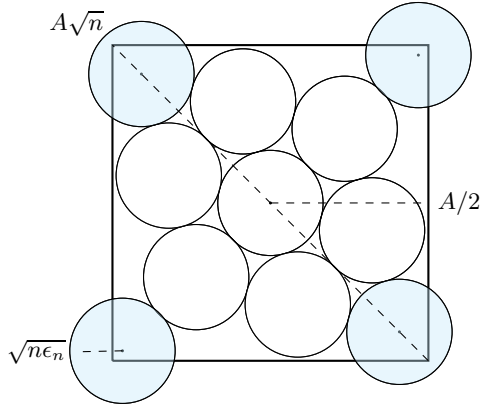
\begin{figure}[H]
    \centering
	\scalebox{1}{
\begin{tikzpicture}[scale=.7][thick]
%\draw[thick] (0,0) circle (3.1cm);
\foreach \s in {3}
{
\draw [thick] (-\s,-\s) -- (\s,-\s) -- (\s,\s) -- (-\s,\s) -- (-\s,-\s);
}
%% Entire Spheres
\draw (-1.41,-1.41) circle (1cm);
\draw [fill=white, fill opacity=0.3] (-1.41,-1.41) circle (1cm);
\draw (0,0) circle (1cm);
\draw [] (0,0) circle (1cm);
\draw (1.41,1.41) circle (1cm);
\draw [] (1.41,1.41) circle (1cm);
% ---
\draw (+.52,-1.93) circle (1cm);
\draw [fill=white, fill opacity=0.5] (+.52,-1.93) circle (1cm);
\draw (1.93,-.52) circle (1cm);
\draw [] (1.93,-.52) circle (1cm);
% ---
\draw (-1.93,+.52) circle (1cm);
\draw [fill=white, fill opacity=0.5] (-1.93,.52) circle (1cm);
\draw (-.52,1.93) circle (1cm);
\draw [] (-.52,1.93) circle (1cm);
% --- Partial Spheres
\node [fill=black, shape=circle, inner sep=.4pt] ($.$) at (2.45,-2.45)  {};
\draw (2.45,-2.45) circle (1cm);
\draw [fill=cyan!20!white, fill opacity=0.4] (2.45,-2.45) circle (1cm);
\node [fill=black, shape=circle, inner sep=.4pt] ($.$) at (2.81,2.81)  {};
\draw (2.81,2.81) circle (1cm);
\draw [fill=cyan!20!white, fill opacity=0.4] (2.81,2.81) circle (1cm);
\node [fill=black, shape=circle, inner sep=.4pt] ($.$) at (-2.45,2.45) {};
\draw (-2.45,2.45) circle (1cm);
\draw [fill=cyan!20!white, fill opacity=0.4] (-2.45,2.45) circle (1cm);
\node [fill=black, shape=circle, inner sep=.4pt] ($.$) at (-2.81,-2.81) {};
\draw (-2.81,-2.81) circle (1cm);
\draw [fill=cyan!20!white, fill opacity=0.4] (-2.81,-2.81) circle (1cm);
\node [fill=black, shape=circle, inner sep=.4pt] ($.$) at (0,0) {};
% --- Arrow
\draw [dashed] (0,0) -- (3,0) node [right,font=\small] {$A/2$};
\draw [dashed] (-2.81,-2.81) -- (-3.71,-2.83) node [left,font=\small] {$\sqrt{n\epsilon_n}$};
\draw [dashed] (3,-3) -- (-3,3) node [above left,font=\small] {$A\sqrt{n}$};
\end{tikzpicture}}
	\caption{Illustration of a saturated sphere packing inside a cube, where small spheres of radius $r_0 = \sqrt{n\epsilon_n}$ cover a larger cube of edge $r_1 = A$ and diameter $A\sqrt{n}$. The small spheres are disjoint from each other and have a non-empty intersection with the large cube. Some of the small spheres, marked in cyan, are not entirely contained within the larger cube, and yet they are considered to be a part of the packing arrangement with their center inside the cube. As we assign a codeword to each small sphere center, the $1$-norm and arithmetic mean of a codeword is bounded by $A$ as required.}
	\label{Fig.Density}
\end{figure}
% ---
In particular, we use a packing argument that has a similar flavor as that observed in the Minkowski--Hlawka theorem in lattice theory \cite{CHSN13}. We use the property that there exists an arrangement $\bigcup_{i=1}^{L} \S_{\fu_i}(n,\sqrt{n\epsilon_n})$ of non-overlapping spheres inside a $\Q_{\f0}(n,A)$, with a density of $\Delta_n(\mathscr{S})\geq 2^{-n}$ \cite[Lem.~2.1]{C10}.
Specifically, consider a saturated packing arrangement of $L(n,R)=2^{(n\log n)R}$ spheres with radius $r_0=\sqrt{n\epsilon_n}$ covering the larger cube $\Q_{\f0}(n,A)$ with edge $r_1=A$, i.e., such that no spheres can be added without overlap. Then, for such an arrangement, there cannot be a point in the larger cube $\Q_{\f0}(n,A)$ with a distance of more than $2r_0$ from all sphere centers. Otherwise, a new sphere could be added. As a consequence, if we double the radius of each sphere, the $2r_0$-radius spheres cover the whole hypercube of edge $r_1$. In general, the volume of a hypersphere of radius $r$ is given by
%%%
\begin{align}
    \text{Vol}\left(\S_{\fx}(n,r)\right) = \frac{\pi^{\frac{n}{2}}}{\Gamma(\frac{n}{2}+1)} \cdot r^{n} \,,\,
    \label{Eq.VolS}
\end{align}
%%%
(\cite[see Eq.~(16)]{CHSN13}).
%  and the following bound for logarithm of the volume holds
% %
% \begin{align}
%     \log\left(\text{Vol}\left(\S_{\fx}(n,r)\right)\right) < -\frac{n}{2}\log(\frac{n}{2\pi e})-\frac{1}{2}\log n\pi +
%     n\log r \,,\,
%     \label{Eq.Log_Vol}
% \end{align}
% %
% % where $0 < \tau < \frac{\log e}{6n}$
% \cite[See Eq.~(18)]{CHSN13}.
Hence, doubling the radius multiplies the volume by $2^n$. Since the $2r_0$-radius spheres cover the entire hypercube of edge $r_1$, it follows that the original $r_0$-radius packing has a density of at least $2^{ -n}$, i.e.,
%%%
\begin{align}
    \Delta(\mathscr{S})\geq 2^{-n}.
    \label{Eq.MinkowskiDeltaFast}
\end{align}
%%%
We assign a codeword to the center $\fu_i$ of each small sphere.
The codewords satisfy the input constraint as
% ---
\begin{align*}
    0 < u_{i,t} \leq r_1 = A \,,\,
\end{align*}
% ---
for all $t \in [\![n]\!]$ and $i\in [\![L(n,R)]\!]$ which is equivalent to
% ---
\begin{align}
    \label{Ineq.Norm_Infinity}
    \norm{\fu_i}_{\infty} \leq A \;.\,
\end{align}
% ---
Since all the small spheres have volume equal to $\text{Vol}(\S_{\fu_1}(n,r_0))$ and cover the whole cube, the total number of spheres is bounded from below by
%---
\begin{align}
    L&=\frac{\text{Vol}\left(\bigcup_{i=1}^{L}\S_{\fu_i}(n,r_0\right)}{\text{Vol}(\S_{\fu_1}(n,r_0))}
    \nonumber\\&
    \geq\frac{\text{Vol}\left(\Q_{\f0}(n,A)\cap\bigcup_{i=1}^{L}\S_{\fu_i}(n,r_0)\right)}{\text{Vol}(\S_{\fu_1}(n,r_0))}
    \nonumber\\&
    =\frac{\Delta(\mathscr{S})\cdot
    \text{Vol}\left[\Q_{\f0}(n,A)\right]}{\text{Vol}(\S_{\fu_1}(n,\sqrt{n\epsilon_n}))}
    \nonumber\\&
    \geq 2^{-n}\cdot \frac{\text{Vol}\left[\Q_{\f0}(n,A)\right]}{\text{Vol}(\S_{\fu_1}(n,\sqrt{n\epsilon_n}))}
    \nonumber\\
    & = 2^{-n}\cdot \frac{A^n}{\text{Vol}(\S_{\fu_1}(n,\sqrt{n\epsilon_n}))} \,,\,
\end{align}
% ---
where the second equality is due to (\ref{Eq.DensitySphereFast}) and the inequality that follows holds by (\ref{Eq.MinkowskiDeltaFast}). Hence,
\begin{align}
    2^{(n\log n)R} \geq 2^{-n}\cdot \frac{A^n}{\text{Vol}(\S_{\fu_1}(n,\sqrt{n\epsilon_n}))} \,.\,
\end{align}
Thus,
% ---
\begin{align}
    R & \geq \frac{1}{n\log n} \cdot \log \left( \frac{A^n}{\text{Vol}\left(\S_{\fu_1}(n,r_0)\right)} \right)
    \nonumber\\
    & = \frac{1}{n\log n} \left[ n\log\left( \frac{A}{\sqrt{\pi} r_0 } \right)+\log \left(\frac{n}{2}! \right) \right]
    \nonumber\\
    & = \frac{1}{n\log n} \left[ o(n\log n) +  \frac{n}{2}\log n - n \log r_0 \right] \,.\,
\end{align}
% ---
Hence, for $r_0 = \sqrt{n\epsilon_n}$, we obtain
% ---
\begin{align*}
    R \geq \frac{1}{n\log n} \left[ o(n\log n) +  \frac{n}{2}\log n - \frac{1}{4}(1+b) \cdot n \log n \right] \,,\,
\end{align*}
% ---
which tends to $\frac{1}{4}$ when $n \to \infty$ and $b\rightarrow 0$.
% ------------
\begin{remark}
\label{Rem.Positive_Rate}
%To explain how we get a positive rate, we compare between the volume of the small hypersphere and the larger hypercube.
We note that the ratio of the small spheres in our construction grows with $n$, as $\sim n^{\frac{1}{4}}$. 
It is well-known that the volume of an $n$-dimensional \emph{unit}-hypersphere, i.e., with a radius of $r_0=1$, tends to zero, as $n\to\infty$ \cite[Ch.~1, Eq.~(18)]{CHSN13}. 
Nonetheless, %for a radius that does not equal to one, 
we observe that the volume still tends to zero for a radius $r_0 = n^c$, where  $0 < c < \frac{1}{2}$. In particular,
% ---
\begin{align}
    & \lim_{n\to\infty} \text{Vol}\left(\S_{\fu_1}(n,r_0)\right) = \lim_{n\to\infty}\frac{\pi^{\frac{n}{2}}}{\Gamma(\frac{n}{2}+1)}\cdot r_0^n
    \nonumber\\
    & = \lim_{n\to\infty} \frac{\pi^{\frac{n}{2}}}{\frac{n}{2}!}\cdot r_0^n
    % \nonumber\\
    % & = \lim_{n\to\infty} \frac{(\sqrt{\pi}r_0)^n}{e^{\frac{n}{2}\ln \frac{n}{2}}}
    % \nonumber\\
    % &
    = \lim_{n\to\infty} \left(\sqrt{\frac{2\pi}{n}}r_0\right)^n \,,\,
    \label{Eq.nSphere_Volume}
\end{align}
% ---
where the last equality follows by Stirling's approximation, 
\begin{align}
    \ln n! = n\ln n - n + o(n) \,.
\end{align} 
Observe that the last expression in (\ref{Eq.nSphere_Volume}) tends to zero, not only for a vanishing radius, but also for $r_0 = n^c$, where $0 < c < \frac{1}{2}$.
Whereas, the volume of the hypercube $\Q_{\f0}(n,A)$ with an edge of length $A$ tends to
% ---
\begin{align}
    \lim_{n\to\infty} \text{Vol}\left[\Q_{\f0}(n,A)\right] = \lim_{n\to\infty} A^n = \begin{cases} 0 & A < 1 \,,\, \\ 1 & A = 1 \,,\, \\ \infty & A > 1 \,.\, \end{cases}
\end{align}
% ---

As for the number of packed spheres, observe that the log-ratio of the volumes satisfies
% ---
\begin{align}
    \label{Eq.Volume_Ratio}
    & \log \left( \frac{\text{Vol}\left[\Q_{\f0}(n,A)\right]}{\text{Vol}\left(\S_{\fu_1}(n,r_0)\right)} \right) = \log \left( \frac{A^n}{\pi^{\frac{n}{2}}{} r_0^n} \cdot \frac{n}{2}!\right)
    \nonumber\\
    & = n\log\left( \frac{A}{\sqrt{\pi} r_0 } \right)+\log \left(\frac{n}{2}! \right)
    \nonumber\\
    & = n\log\left( \frac{A}{\sqrt{\pi}} \right) + \frac{n}{2}\log \frac{n}{2} - \frac{n}{2}\log e - n \log r_0 + o(n)
    % \nonumber\\
    % & = n\log\left( \frac{A}{\sqrt{\pi}} \right) + \frac{1}{2} n\log n - \frac{n}{2} - \frac{n}{2}\log e - n \log r_0 + o(n)
    \nonumber\\
    & = o(n\log n) +  \frac{1}{2} n\log n  - n \log r_0
    % \nonumber\\
    % & = \log \left( \left( \frac{A}{\sqrt{\pi} r_0} \right) ^{n} \cdot \left( \frac{n}{2} \ln \frac{n}{2} - \frac{n}{2} + o(n) \right) \right)
    % % \nonumber\\
    % % & = n \log\left(\frac{A}{\sqrt{\pi}}\right) - n\log r_0 + \log \left( \frac{n}{2} \ln \frac{n}{2} - \frac{n}{2} + o(n) \right)
    % % \nonumber\\
    % % & = n \log\left(\frac{A}{\sqrt{\pi}}\right) - n\log r_0 + \log \left( \frac{n}{2} \left[ \ln \frac{n}{2} - 1 + o(n) \right] \right)
    % \nonumber\\
    % & = n \log\left(\frac{A}{\sqrt{\pi}}\right) - n\log r_0 + \log \left( \frac{n}{2} \right) + o(n) %+ \log \left(\ln \frac{n}{2} - 1 + o(n) \right)
    % \nonumber\\
    % & = n \log\left(\frac{A}{\sqrt{\pi}}\right) - n\log r_0 + o(n)
    % % \nonumber\\
    % % & = n \log\left(\frac{A}{\sqrt{\pi}}\right) - n\log r_0 + o(n)
    % % \nonumber\\
    % % & \sim n \log\left(\frac{A}{\sqrt{\pi}}\right) - n\log r_0 + \log \left( n \ln n - n + o(n) \right)
    % % \nonumber\\
    % % & \sim \log(n\ln n) -n\log (r_0)
    \,.\,
\end{align}
% ---
%We note that logarithm of the ratio of volume given in (\ref{Eq.Volume_Ratio}) for a coding scale of $n\log n$ is $(n\log n)R$, 
Hence, for  $r_0 = n^c$, % with $0 < c < \frac{1}{2}$, #
we obtain
% ---
\begin{align}
    R & \geq \frac{1}{n\log n} \log \left( \frac{\text{Vol}\left[Q_{\f0}(n,A)\right]}{\text{Vol}\left(\S_{\fu_1}(n,r_0)\right)} \right)
    \nonumber\\
    & = \frac{1}{n\log n} \left[\frac{1}{2} n\log n  - c \, n \log n  + o(n\log n) \right] \,,\,
\end{align}
% ---
which tends to $\frac{1}{2} - c$ for every value of $A$.

% We consider the case $A < 1$ specifically as follows: observe that in this case volume of the hypercube does not grows in size, still one can obtain positive rate. The rationale for this counter-intuitive phenomenon is that growth of the decreasing volume for hypercube in these case is much less compared to that of the volume for hypersphere and remains manageable as indicated by denominator in \ref{Eq.Volume_Ratio}. Thus, logarithm of ratio of the volumes diverges to infinity as of the form
% % ---
% \begin{align}
%     \log \left( \frac{\text{Vol}\left[Q_{\f0}(n,A)\right]}{\text{Vol}\left(\S_{\fu_1}(n,r_0)\right)} \right) & \sim \log \left( \frac{A}{n^{c-\frac{1}{2}}} \right)^n
%     \nonumber\\
%     & = \log \left( \frac{n^{\frac{1}{2}-c}}{A} \right)^n
%     \nonumber\\
%     & = \log \left( Bn^{\frac{1}{2}-c}\right)^n
%     \nonumber\\
%     & \sim \left( \frac{1}{2} -c \right) \cdot n \log n
% \end{align}
% % ---
% which tends to $\frac{1}{2} - c$ where $B = \frac{1}{A} >1$.

% In brief, the rationale for phenomenon of obtaining a positive rate with diverging ($\to \infty$) radius of hypersphere is that the volume 
% of small hyperspheres grow in a manageable size such that its volume remains still negligible in higher dimension compared to that of the ambient hypercube.
\end{remark}

\subsubsection*{Encoding}
Given a message $i\in [\![L(n,R)]\!]$, transmit $\fx=\fu_i$.
\subsubsection*{Decoding}
Let
% ---
\begin{align}
    \delta_n = \frac{A}{3n^{\frac{1}{2}(1-b)}} \,.\,
\end{align}
% ---
To identify whether a message $j\in \M$ was sent, the decoder checks whether the channel output $\mathbf{y}$ belongs to the following decoding set, % $\D_j$ or not, where
% \begin{align}
%     \left|\sum_t (Y_t-(u_{j,t}+\lambda))^2 -
% \left(\sum_t u_{j,t}+n\lambda\right) \right|\leq n\delta_n
% \end{align}
%
%%%
\begin{align}
    \D_j & = \left\{ \fy \in \Y^n \;:\, \left| \frac{1}{n} \sum_{t=1}^n \left[ \left( y_t - \left( u_{j,t} + \lambda \right) \right)^2 - \left( \lambda + u_{j,t} \right) \right] \right| \leq \delta_n \right\} \,.\,
    %\nonumber\\
    \label{Eq.Decoding_Set}
    % &= \left\{\fy\in\mathbb{N}_0^n \,:\; \sum_{t=1}^n \left(y_t - (\lambda + u_{j,t})\right)^2
    % \leq n(\delta_n+ \lambda) + \sum_{t=1}^n u_{j,t} \right\}
\end{align}
%%%
\subsubsection*{Error Analysis}
Consider the type I error, i.e., when the transmitter sends $\fu_i$, yet $\fY\notin\D_i$. For every $i\in[\![L(n,R)]\!]$, the type I error probability is bounded by
% ---
\begin{align}
    P_{e,1}(i) & = \Pr \left( \left| \frac{1}{n}\sum_{t=1}^n \left[ \left( Y_t - \left( u_{i,t} + \lambda \right) \right)^2 - \left( \lambda + u_{i,t} \right) \right] \right| > \delta_n \, \Big| \, \fx = \fu_i \right) \;.\,
    \label{Eq.TypeIError}
\end{align}
% ---
Let $Y_t(i)\sim\text{Pois}(\lambda+u_{i,t})$ denote the channel output at time $t$ given that $\fx=\fu_i$.
Now, we derive the expectation as follows,
\begin{align}
     \mathbb{E} \left\{ \frac{1}{n} \sum_{t=1}^n \left[ \left( Y_t(i) - \left( u_{i,t} + \lambda \right) \right)^2 - \left( \lambda + u_{i,t} \right) \right] \right\}
    % & =\frac{1}{n} \sum_{t=1}^n \mathbb{E}\left\{ (Y_t(i)-(u_{i,t}+\lambda))^2 -(\lambda+u_{i,t}) \right\}
    % \nonumber\\
    % &= \frac{1}{n} \sum_{t=1}^n \mathbb{E}\left\{ (Y_t(i)-(u_{i,t}+\lambda))^2\right\} - \mathbb{E}\{\lambda+u_{i,t}\}
    % \nonumber\\
    & = \frac{1}{n} \sum_{t=1}^n \left[ \text{var} \left\{ Y_t(i) \right\} - \left( u_{i,t} + \lambda \right) \right]
    \nonumber\\
    % &= \frac{1}{n} \sum_{t=1}^n \text{var}\left\{ Y_t(i)-(u_{i,t}+\lambda) \right\}
    %  \nonumber\\
    % &= \frac{1}{n} \sum_{t=1}^n \lambda+u_{i,t}
    % \nonumber\\
    % & = \lambda + \frac{1}{n} \sum_{t=1}^n u_{i,t} \leq \lambda + A
    & = 0 \,.\,
    \label{Eq.Expectation}
\end{align}
Hence, the expectation of the random variable within the absolute value in (\ref{Eq.TypeIError}) is zero.
Next, we compute the variance as follows
% ---
\begin{align}
    \text{var}\left\{ \frac{1}{n} \sum_{t=1}^n  \left[ \left( Y_t(i) - \left( u_{i,t} + \lambda \right) \right)^2 - \left( \lambda + u_{i,t} \right) \right] \right\}
    % & \stackrel{(a)}{=} \frac{1}{n} \sum_{t=1}^n \text{var} \left\{ (Y_t(i)-(u_{i,t}+\lambda))^2 -(\lambda+u_{i,t}) \right\}
    %  \nonumber\\
     & = \frac{1}{n^2} \sum_{t=1}^n \text{var} \left\{ \left( Y_t(i) - \left( u_{i,t} + \lambda \right) \right)^2 \right\}
     \nonumber\\
     & \leq \frac{1}{n} \mathbb{E} \left[ \left( Y_t(i) - \left( u_{i,t} + \lambda \right) \right)^4 \right] \,,\,
    \label{Ineq.4thMoment}
\end{align}
% ---
where the equality holds because the channel is memoryless and inequality follows by $\text{var}\{Z\} \leq \mathbb{E}\{Z^2\}$. The moment-generating function (MGF) of a Poisson variable
$Z\sim\text{Pois}(\lambda_Z)$ is given by $G_Z(\alpha) = e^{\lambda_Z(e^{\alpha}-1)}$. Hence, 
for $X=Z-\lambda_Z$, MGF turns to be $G_X(\alpha) = e^{\lambda_Z (e^{\alpha} - 1 - \alpha)}$. Therefore, we have 
% ---
\begin{align}
    \mathbb{E}\{ X^4 \} & =
    \frac{d^4}{d\alpha^4}G_X(\alpha)
    \bigg|_{\alpha=0}\nonumber\\
    & = \lambda_Z \left( \lambda_Z^3 e^{3\alpha} + 6\lambda_Z^2 e^{2\alpha} + 7\lambda_Z e^{\alpha} + 1 \right) e^{\alpha + \lambda_Z e^{\alpha} - \lambda_Z}
    \bigg|_{\alpha = 0} \nonumber\\
    & = \lambda_Z^4 + 6 \lambda_Z^3 + 7\lambda_Z^2 + \lambda_Z 
    \nonumber\\
    & \leq 7\lambda_Z^4 + 7\lambda_Z^3 + 7\lambda_Z^2 + 7\lambda_Z
    \nonumber\\
    & = 7 \left( \lambda_Z^4 + \lambda_Z^3 + \lambda_Z^2 + \lambda_Z \right)
    \nonumber\\
    & = 7 \left( \left( \lambda + u_{i,t} \right)^4 + \left( \lambda + u_{i,t} \right)^3 + \left( \lambda + u_{i,t} \right)^2 + \left( \lambda + u_{i,t} \right) \right)
    \nonumber\\
    & \leq 7 \left( \left( \lambda + A \right)^4 + \left( \lambda + A \right)^3 + \left( \lambda + A \right)^2 + \left( \lambda + A \right) \right)
    \,.\,
\end{align}
% ---
% we have
% $G_X(\alpha) = \lambda e^{\lambda(e^{\alpha}-1+\alpha)}$;
% $G_X(0) = 1$, $G_X'(0) = \lambda$, $G_X''(0) = 2\lambda^2$, $G_X'''(0) = 4\lambda + 1$, and
% \begin{align}
%     \mathbb{E}\{ X^4 \}&=
%     %
%     \frac{d^4}{d\alpha^4}G_X(\alpha)
%     \bigg|_{\alpha=0}\nonumber\\
%     &= \lambda[e^{\alpha}-1]G'''_X(\alpha) + [3\lambda e^{\alpha}]G''_X(\alpha) + [3\lambda e^{\alpha}]G'_X(\alpha) + [\lambda e^{\alpha}]G_X(\alpha)
%     \bigg|_{\alpha=0} \nonumber\\&
%     = 3\lambda^2+\lambda \,.\,
% \end{align}
%
Thereby,
% ---
\begin{align}
    \text{var}\left\{ \frac{1}{n}\sum_{t=1}^n \left[ \left( Y_t(i) - \left( u_{i,t} + \lambda \right) \right)^2 - \left( \lambda + u_{i,t} \right) \right] \right\} \leq \frac{7}{n} \left( \left( \lambda + A \right)^4 + \left( \lambda + A \right)^3 + \left( \lambda + A \right)^2 + \left( \lambda + A \right) \right) \,.\,
\end{align}
% ---

Now by (\ref{Eq.Expectation}) and (\ref{Ineq.4thMoment}), we can rewrite the type I error given in (\ref{Eq.TypeIError}) as follows
\begin{align}
    P_{e,1}(i) &= \Pr \left( \left| \frac{1}{n} \sum_{t=1}^n \left( Y_t(i) - \left( u_{i,t} + \lambda \right) \right)^2 - \left( \lambda + u_{i,t} \right) \right| > \delta_n \right)
    \label{Ineq.TypeIError}
     \nonumber\\
     & \leq \frac{7 \left( \left( \lambda + A \right)^4 + \left( \lambda + A \right)^3 + \left( \lambda + A \right)^2 + \left( \lambda + A \right) \right)}{n\delta_n^2}
    % \nonumber\\
    %  & = \frac{3\lambda^2+\lambda}{n(\frac{A}{3n^{\frac{1}{2}(1-b)}})^2}
    \nonumber\\
     & = \frac{63 \left( \left( \lambda + A \right)^4 + \left( \lambda + A \right)^3 + \left( \lambda + A \right)^2 + \left( \lambda + A \right) \right)}{A^2 n^b}
     \nonumber\\
     & \leq \lambda_1 \,,\,
\end{align}
% ---
for sufficiently large $n$ and arbitrarily small $\lambda_1>0$, where the first inequality follows by Chebyshev's inequality.

% ---
Next, we address type II errors, i.e., when $\fY\in\D_j$ while the transmitter sent $\fu_i$.
Then, for every $i,j\in[\![L(n,R)]\!]$, where $i\neq j$, the type II error probability is given by
% ---
\begin{align}
    P_{e,2}(i,j) = \Pr \left( \left| \frac{1}{n}\sum_{t=1}^n \left[ \left( Y_t(i) - \left( u_{j,t} + \lambda \right) \right)^2 - \left( \lambda + u_{j,t} \right) \right] \right| \leq \delta_n \right) \;.\,
    \label{Eq.Pe2G}
\end{align}
% ---
Then,
% ---
\begin{align}
    P_{e,2}(i,j) = \Pr\left( \left| \frac{1}{n} \sum_{t=1}^n \left[ \left( Y_t(i) - \left( u_{i,t} + \lambda \right) + \left( u_{i,t} - u_{j,t} \right) \right)^2 - \left( \lambda + u_{j,t} \right) \right] \right| \leq \delta_n \right) \;.\,
\end{align}
% ---
Observe that the squared argument inside the absolute value can be expressed as
% ---
\begin{align}
    & \sum_{t=1}^n \left( Y_t(i) - \left( \lambda + u_{i,t} \right) + \left( u_{i,t} - u_{j,t} \right) \right)^2
    \nonumber\\
    & = \norm{ \fY(i) - \left( \lambda \boldsymbol{1} + \fu_i \right)}^2 + \norm{\fu_i - \fu_j}^2 + 2\sum_{t=1}^n \left( u_{i,t} - u_{j,t} \right) \left( Y_t(i) - \left( \lambda + u_{i,t} \right) \right) \;,\,
     \label{Eq.Pe2norm}
\end{align}
% ---
where $\mathbf{1} = \left(1,1, \cdots,1 \right)$. Then, define the event
% ---
\begin{align}
    \E_0 = \left\{\left| \sum_{t=1}^n \left( u_{i,t} - u_{j,t} \right) \left( Y_t(i) - \left( \lambda + u_{i,t} \right) \right) \right| > \frac{n\delta_n}{2} \right\} \;.\,
\end{align}
% ---
By Chebyshev's inequality, the probability of this event vanishes,
%%%
\begin{align}
    \Pr(\E_0) &\leq 
    % \frac{\text{var}\{\sum_{t=1}^n (u_{i,t}-u_{j,t})\cdot Y_t(i)\}}{(\frac{n\delta_n}{2})^2}
    % % \nonumber\\
    % % &=\frac{4\sum_{t=1}^n\text{var}\{ (u_{i,t}-u_{j,t})\cdot(\text{Pois}(\lambda+u_{i,t})-(\lambda+u_{i,t}))\}}{\beta_n^2}
    %\nonumber\\
    %&=
    \frac{4\sum_{t=1}^n (u_{i,t}-u_{j,t})^2\cdot\text{var}\{Y_t(i))%-(\lambda+u_{i,t})
    \}}{n^2\delta_n^2}
    \nonumber\\
    & = \frac{4\sum_{t=1}^n(u_{i,t}-u_{j,t})^2\cdot(\lambda+u_{i,t})}{n^2\delta_n^2}
    \nonumber\\
    % &\leq
    % \frac{4\sum_{t=1}^n(u_{i,t}-u_{j,t})^2\cdot(\lambda+A)}{\delta_n^2}
    % \nonumber\\
    & \leq \frac{4(\lambda+A)\sum_{t=1}^n(u_{i,t}-u_{j,t})^2}{n^2\delta_n^2}
    \nonumber\\
    & = \frac{4(\lambda+A)\norm{\fu_i-\fu_j}^2}{n^2\delta_n^2} \,.\,
    \label{Eq.PeE0G}
\end{align}
% --- 
Observe that
% ---
\begin{align}
    \norm{\fu_i - \fu_j}^2 & \stackrel{(a)}{\leq} \left(\norm{\fu_i} + \norm{\fu_j}\right)^2
    \nonumber\\
    & \stackrel{(b)}{\leq} \left(\sqrt{n} \norm{\fu_i}_{\infty} + \sqrt{n} \norm{\fu_j}_{\infty} \right)^2
    \nonumber\\
    & \stackrel{(c)}{\leq} \left(\sqrt{n} A + \sqrt{n} A \right)^2 
    \nonumber\\
    & = 4nA^2 \;,\
\end{align}
% ---
where $(a)$ holds by the triangle inequality, $(b)$ follows since $\norm{.} \leq \sqrt{n} \norm{.}_{\infty}$ and $(c)$ is valid by (\ref{Ineq.Norm_Infinity}). Hence
% ---
\begin{align}
    \Pr(\E_0) & \leq \frac{4(\lambda+A)\norm{\fu_i-\fu_j}^2}{n^2\delta_n^2}
    \nonumber\\
    & \leq \frac{16n(\lambda+A)A^2}{n^2\delta_n^2}
    \nonumber\\
    & = \frac{16(\lambda+A)A^2}{n\delta_n^2} 
    \nonumber\\
    & = \frac{144(\lambda+A)}{n^b}
    \nonumber\\
    & \leq \zeta_1 \,,\,
\end{align}
% ---
for sufficiently large $n$, where $\zeta_1 > 0$ is arbitrarily small. Furthermore, observe that given the complementary event $\E_0^c$, we have
\begin{align}
    \label{Eq.E_0_Result}
    %\E_0^c(\text{result}) \triangleq 
    2\sum_{t=1}^n \left( u_{i,t} - u_{j,t} \right) \left( Y_t(i) - \left( \lambda + u_{i,t} \right) \right) \geq -n\delta_n \,.\,
\end{align}
% ---
Therefore, the event $\E_0^c$, the type II error event in (\ref{Eq.Pe2G}), and the identity in (\ref{Eq.Pe2norm}) together imply that the following event occurs,
% ---
\begin{align}
% &\norm{\fY(i)-(\lambda+\fu_i)+(\fu_i-\fu_j)}^2 \leq \delta_n
% \nonumber\\
%\Rightarrow
\E_1
&=\left\{ \norm{\fY(i)-(\lambda\boldsymbol{1}+\fu_i)}^2 + \norm{\fu_i-\fu_j}^2 \leq 2n\delta_n \right\} \,.\,
\end{align}
% ---
Applying the law of total probability, we have
%by (\ref{Eq.Pe2G})-(\ref{Eq.PeE0G}),
\begin{align}
    P_{e,2}(i,j) & \stackrel{(a)}{=}
   \Pr \left( \left\{ \left| \frac{1}{n} \sum_{t=1}^n \left[ \left( Y_t(i) - \left( u_{i,t} + \lambda \right) + \left( u_{i,t} - u_{j,t} \right) \right)^2 - \left( \lambda + u_{j,t} \right) \right] \right| \leq \delta_n \right\} \cap \E_0 \right)
    \nonumber\\
    & + \Pr \left( \left\{ \left| \frac{1}{n}\sum_{t=1}^n \left[ \left( Y_t(i) - \left( u_{i,t} + \lambda \right) + \left( u_{i,t} - u_{j,t} \right) \right)^2 - \left( \lambda + u_{j,t} \right) \right] \right| \leq \delta_n \right\} \cap \E_0^c \right)
    \nonumber\\
    % &\stackrel{(b)}{\leq}
    % \Pr(\E_0)+ \Pr\left(\Big| \frac{1}{n}\sum_{t=1}^n \big[ (Y_t(i)-(u_{i,t}+\lambda)+u_{i,t}-u_{j,t})^2 -(\lambda+u_{j,t}) \big] \Big| \leq \delta_n \cap \E_1 \right)
    % \nonumber\\
    % &\stackrel{(c)}{\leq}
    % \Pr(\E_0)+\Pr\left( \norm{\text{Pois}(\lambda+\fu_i)-(\lambda+\fu_i)}^2 + \norm{\fu_i-\fu_j}^2 \leq n(\delta_n+ \lambda+\beta_n) + \sum_{t=1}^n u_{i,t} \right)
    % \nonumber\\
    &\stackrel{(b)}{\leq}
    \zeta_1 + \Pr\left(\E_1 \right) \,,\,
    %\Pr\left( \norm{\fY(i)-(\lambda\boldsymbol{1} + \fu_i)}^2 + \norm{\fu_i-\fu_j}^2 \leq 2n\delta_n \right) \,,\,
    \label{Eq.Pe2_P_Expanded}
\end{align}
%%%
where $(a)$ is due to (\ref{Eq.Pe2G}) and $(b)$ holds since each probability is bounded by $1$.

Based on the codebook construction, each codeword is surrounded by a sphere of radius $\sqrt{n\epsilon_n}$ thus
\begin{align}
     n\epsilon_n \leq \norm{\fu_i-\fu_j}^2 \,.\,
\end{align}
Thus, reconsidering (\ref{Eq.Pe2_P_Expanded}) we obtain the type II error upper-bounded as follows
%%%
\begin{align}
    \label{Eq.TypeIIErrorAnalysis}
    P_{e,2}(i,j)
    & \leq
    \Pr\left( \norm{\fY(i) - \left( \lambda \boldsymbol{1} + \fu_i \right)}^2 \leq n \left( 2\delta_n - \epsilon_n \right) \right) + \zeta_1
    \nonumber\\
    & =
    \Pr\left( \norm{\fY(i) - \left( \lambda \boldsymbol{1} + \fu_i \right)}^2 \leq -\frac{n\epsilon_n}{3} \right) + \zeta_1
    \nonumber\\
    & = 0 + \zeta_1 \,,\,
 \end{align}
%%%
where the first equality holds since $\delta_n = \frac{\epsilon_n}{3}$.

We have thus shown that for every $\lambda_1,\lambda_2>0$ and sufficiently large $n$, there exists an $(L(n,R), n, \lambda_1, \lambda_2)$ code.
%
% The proof follows by taking the limits $n\rightarrow\infty$, and $b\rightarrow 0$ in (\ref{Eq.Rate}).
\end{proof}
%%%%%%%%%%%%%%%%%%%%%%%%%%%%%%%%%%%%%%%%%%%%%%%%%%%%%%%%%%%%%%%%%
\subsection{Upper Bound (Converse Proof for Th.~\ref{Th.PDICapacity})}
\label{Subsec.ConvFast}
The derivation for the DTPC is more involved than the Gaussian derivation as in \cite{Salariseddigh_arXiv_ITW}. Instead of establishing a minimum distance between the codewords as before, we use the letter-wise ratio.

We show that the capacity is bounded by $\mathbb{C}_{DI}(\W,L)\leq \frac{3}{2}$. In our previous work on fading channels \cite{Salariseddigh_ITW,Salariseddigh_arXiv_ITW}, the converse proof was based on establishing a minimum distance between each pair of codewords. Here, on the other hand, we use the stronger requirement that the letter-wise ratio for each pair is distanced from $1$.

Suppose that $R$ is an achievable rate in the $L$-scale for the Poisson channel. Consider a sequence of $(L(n,R), n, \lambda_1, \allowbreak \lambda_2)$ codes $(\U^{(n)},\D^{(n)})$ such that $\lambda_1^{(n)}$ and $\lambda_2^{(n)}$ tend to zero as $n\rightarrow\infty$. Given a codebook $\U^{(n)}=\{\fu_i\}_{i\in [\![ L(n,R) ]\!]}$, we define the shifted codewords $\fv_i$ by
% ---
\begin{align}
    v_{i,t} = \lambda + u_{i,t} \;,\,
\end{align}
% ---
for $t \in [\![ n ]\!]$.

We begin with the following lemma on the letter-wise ratio for every pair of codewords.
%%%
\begin{lemma}
\label{Lem.DConverseFast}
Consider a sequence of codes  as described above. Then, given a sufficiently large $n$, the codebook 
$\U^{(n)}$ satisfies the following property.
For every pair of codewords, $\fv_{i_1}$ and $\fv_{i_2}$, there exists at least one letter $t \in [\![n]\!]$ such that 
% the ratio $\frac{v_{i_2,t}}{v_{i_1,t}}$ for corresponding letters of the two codewords is distanced from $1$ by at least $\epsilon'_n$. That is,
% ---
\begin{align}
    \label{Eq.Converse_Lem}
    \left|1-\frac{v_{i_2,t}}{v_{i_1,t}}\right| > \epsilon'_n \,,\,
\end{align}
% --------------------
for all $i_1,i_2\in [\![L(n,R)]\!]$, such that $i_1\neq i_2$,
with
\begin{align}
\label{Eq.epsilonn_p}
  \epsilon'_n = \frac{P_{\,\text{max}}}{n^{1+b}} \,,\,
\end{align}
%%%
where $b>0$ is arbitrarily small.
\end{lemma}
% ---
\begin{proof}[Converse Proof]
    Fix $\lambda_1$ and $\lambda_2$. Let $\kappa, \delta > 0$ be arbitrarily small.
    Assume to the contrary that 
    there exist two messages $i_1$ and $i_2$, where $i_1\neq i_2$, such that
    % ---
    \begin{align}
        \label{Ineq.Converse_Lem_Complement}
        \left|1-\frac{v_{i_2,t}}{v_{i_1,t}}\right| \leq \epsilon'_n \;,\,
    \end{align}
    % ---
    for all $t\in[\![n]\!]$.

    In order to show contradiction, we will bound the sum of the two error probabilities, $P_{e,1}(i_1)+P_{e,2}(i_2,i_1)$, from below. To this end, define 
   \begin{align}
       \B_{i_1} = \left\{\fy \in \D_{i_1} \,:\,
       \frac{1}{n}\sum_{t=1}^n y_t \leq  \lambda + P_{\,\text{max}} + \delta \right\} \;.\,
   \end{align}
    Then, observe that
    % ---
    \begin{align}
    \label{Eq.Error_Sum_1}
    P_{e,1}(i_1)+P_{e,2}(i_2,i_1)
    & = 1- \sum_{\fy\in\D_{i_1}} W^n \left( \fy \, \big| \, \fu_{i_1} \right) + \sum_{\fy\in\D_{i_1}} W^n \left( \fy \, \big| \, \fu_{i_2} \right)
    \nonumber\\
    & \geq 1- \sum_{\fy\in\D_{i_1}} W^n \left( \fy \, \big| \, \fu_{i_1} \right) + \sum_{\fy\in\D_{i_1} \cap \B_{i_1}} W^n \left( \fy \, \big| \, \fu_{i_2} \right) \,.\,
    \end{align}
    % ---
    
    Now, consider the first sum in (\ref{Eq.Error_Sum_1}),
    % ---
    \begin{align}
        \label{Ineq.Error_I_Complement}
        \sum_{\fy\in\D_{i_1}} W^n \left( \fy \, \big| \, \fu_{i_1} \right) & = \sum_{\fy\in\D_{i_1}\cap\B_{i_1}} W^n \left( \fy \, \big| \, \fu_{i_1} \right) + \sum_{\fy \in \D_{i_1}\cap\B_{i_1}^c} W^n \left( \fy \, \big| \, \fu_{i_1} \right)
        \nonumber\\
        & \leq \sum_{\fy \in \D_{i_1}\cap\B_{i_1}} W^n \left( \fy \, \big| \, \fu_{i_1} \right) + \Pr\left( \frac{1}{n} \sum_{t=1}^n Y_t > \lambda + P_{\,\text{max}} + \delta \, \bigg| \, \fx = \fu_{i_1} \right)
    \end{align}
    % ---
    Now we focus on the probability in the right hand side. By subtracting expectation from each term inside the argument and applying the Chebyschev's inequality, we can bound this probability from above as follows
    % ---
    \begin{align}
        & \Pr \left( \frac{1}{n} \sum_{t=1}^n Y_t -
        \left[ \lambda + \frac{1}{n} \sum_{t=1}^n u_{i_1,t}
        \right] > P_{\,\text{max}} + \delta - \frac{1}{n} \sum_{t=1}^n u_{i_1,t} \, \bigg| \, \fx = \fu_{i_1} \right)
         \leq \frac{\text{var} \left[ n^{-1} \sum_{t=1}^n Y_t \, \big| \, \fx =\fu_{i_1} \right] }{\left( P_{\,\text{max}} + \delta - n^{-1} \sum_{t=1}^n u_{i_1,t} \right)^2}
        \;.\,
        \label{Ineq.ErrorI_Complement1}
   \end{align}
% ---
Observe that due to the peak input constraint, $P_{\max}-n^{-1}\sum_{t=1}^n u_{i_1,t}\geq 0$. Hence, the denominator is bounded from below by
% ---
\begin{align}
    \left(P_{\,\text{max}} + \delta - \frac{1}{n}\sum_{t=1}^n u_{i_1,t}\right)^2\geq \delta^2 \,.
    \label{Ineq.ErrorI_Complement2}
\end{align}
Furthermore, the nominator satisfies
\begin{align}
  \text{var} \left[ \frac{1}{n} \sum_{t=1}^n Y_t~|~\fx=\fu_{i_1} \right]
  &=\frac{1}{n^2} \sum_{t=1}^n\text{var} \left[  Y_t \, \big| \, \fx =\fu_{i_1} \right]
  \nonumber\\
   &=\frac{1}{n^2} \sum_{t=1}^n ( \lambda+ u_{i_1,t} )
   \nonumber\\
   &\leq \frac{1}{n} ( \lambda+ P_{\max} ) \,.
   \label{Ineq.ErrorI_Complement3}
\end{align}
% ---
Therefore, by (\ref{Ineq.ErrorI_Complement1})-(\ref{Ineq.ErrorI_Complement3}),
% ---
\begin{align}
        \Pr \left( \frac{1}{n} \sum_{t=1}^n Y_t -
        \left[ \lambda + \frac{1}{n} \sum_{t=1}^n u_{i_1,t}
        \right] > P_{\,\text{max}} + \delta - \frac{1}{n} \sum_{t=1}^n u_{i_1,t} \, \bigg| \, \fx = \fu_{i_1} \right)
        & \leq \frac{ \lambda + P_{\,\text{max}}}{n\delta^2}
        \nonumber\\
        & \leq \kappa \,. 
\end{align}
% ---
Hence, reconsidering (\ref{Ineq.Error_I_Complement}) we have
% ---
\begin{align}
     \sum_{\fy\in\D_{i_1}} W^n \left( \fy \, \big| \, \fu_{i_1} \right) & \leq \sum_{\fy \in \D_{i_1}\cap\B_{i_1}} W^n \left( \fy \, \big| \, \fu_{i_1} \right) + \kappa \;.\,
\end{align}
% ---
Returning to the sum of error probabilities in (\ref{Eq.Error_Sum_1}), the last bound yields
% ---
\begin{align}
    \label{Eq.Error_Sum_2}
    P_{e,1}(i_1)+P_{e,2}(i_2,i_1)
     & \geq 
    %  1- \Big[\sum_{\fy\in\D_{i_1}\cap\B_{i_1}} W^n(\fy|\fu_{i_1}) + \kappa \Big] + \sum_{\fy\in\D_{i_1}\cap\B_{i_1}}  W^n(\fy|\fu_{i_2})
    % % \nonumber\\
    % % & = 
    1 - \sum_{\fy \in \D_{i_1} \cap \B_{i_1}} \left[ W^n \left( \fy \, \big| \, \fu_{i_1} \right) - W^n \left( \fy \, \big| \, \fu_{i_2} \right) \right] - \kappa \,.\,
\end{align}
% \begin{align}
%     \Pr(\frac{1}{n} \sum_{t=1}^n Y_t > P | \fx = \fu_i) & = \Pr(\frac{1}{n} \sum_{t=1}^n Y_t - (\lambda + \frac{1}{n} \sum_{t=1}^n u_{i,t}) > P - (\lambda + \frac{1}{n} \sum_{t=1}^n u_{i,t}) | \fx = \fu_i)
%     \nonumber\\
%     & = \Pr(\frac{1}{n} \sum_{t=1}^n Y_t - (\lambda + \frac{1}{n} \sum_{t=1}^n u_{i,t}) > A + \delta - \frac{1}{n} \sum_{t=1}^n u_{i,t} | \fx = \fu_i)
%     \nonumber\\
%     & \leq \frac{\text{var}(\frac{1}{n} \sum_{t=1}^n Y_t)}{(A + \delta - \frac{1}{n} \sum_{t=1}^n u_{i,t})^2}
%     \nonumber\\
%     & \leq \frac{\frac{1}{n^2}(n\lambda + \sum_{t=1}^n u_{i,t})}{(A + \delta - \frac{1}{n} \sum_{t=1}^n u_{i,t})^2}
%     \nonumber\\
%     & \leq  \frac{\frac{1}{n^2}(n\lambda + nA)}{\delta^2}
%     \nonumber\\
%     & \leq \frac{\lambda + A}{n\delta^2}
% \end{align}
%%%
    Now let us focus on the summand in the square brackets in (\ref{Eq.Error_Sum_2}).  By the channel law  in (\ref{Eq.Poisson_Channel_Law}),
    % ---
    \begin{align}
        \label{Ineq.Cond_Channel_Diff}
         W^n \left( \fy \, \big| \, \fu_{i_1} \right) - W^n \left( \fy \, \big| \, \fu_{i_2} \right) &=
         \prod_{t=1}^n \frac{e^{-v_{i_1,t}}v_{i_1,t}^{y_t}}{y_t!} - \prod_{t=1}^n \frac{e^{-v_{i_2,t}}v_{i_2,t}^{y_t}}{y_t!} 
        \nonumber\\
        & = e^{-\sum_{t=1}^n v_{i_1,t}} \left[  \prod_{t=1}^n \frac{v_{i_1,t}^{y_t}}{y_t!} - e^{-\sum_{t=1}^n (v_{i_2,t} - v_{i_1,t})} \prod_{t=1}^n \frac{v_{i_2,t}^{y_t}}{y_t!} \right] \,.
    \end{align}
    % ---
    We note that by (\ref{Ineq.Converse_Lem_Complement}), we have $\left| v_{i_2,t} - v_{i_1,t} \right| \leq \epsilon'_n v_{i_1,t}$, which implies
    % ---
    \begin{align}
        1 - \epsilon'_n \leq \frac{v_{i_2,t}}{v_{i_1,t}} \;.\,
    \end{align}
    % ---
    Hence,
    % ---
    \begin{align}
        \label{Ineq.Cond_Channel_Diff_Continuation}
        W^n \left( \fy \, \big| \, \fu_{i_1} \right) - W^n \left( \fy \, \big| \, \fu_{i_2} \right)
        %& \leq e^{-\sum_{t=1}^n v_{i_1,t}} \cdot \prod_{t=1}^n \frac{v_{i_1,t}^{y_t}}{y_t!} \left[ 1 - e^{-\sum_{t=1}^n \epsilon'_n v_{i_1,t}} \prod_{t=1}^n \left(\frac{v_{i_2,t}}{v_{i_1,t}}\right)^{y_t} \right] 
        %\nonumber\\
        & \leq e^{-\sum_{t=1}^n v_{i_1,t}} \cdot \prod_{t=1}^n \frac{v_{i_1,t}^{y_t}}{y_t!} \left[ 1 - e^{-\sum_{t=1}^n \epsilon'_n v_{i_1,t}} \prod_{t=1}^n \left( 1 - \epsilon'_n \right)^{y_t} \right]
        \nonumber\\
        & =  W^n \left( \fy \, \big| \, \fu_{i_1} \right) \left[ 1 - e^{-\sum_{t=1}^n \epsilon'_n v_{i_1,t}} \prod_{t=1}^n \left( 1 - \epsilon'_n \right)^{y_t} \right]\,.
    \end{align}
    % ---
    % where second inequality holds by (\ref{Ineq.Converse_Lem_Complement}). Now in order to retrieve the term $W^n(\fy|\fu_{i_1})$ we distribute back the term $e^{-\sum_{t=1}^n v_{i_1,t}}$ in (\ref{Ineq.Cond_Channel_Diff}) by multiplying it to the product term next to it.
    This can then be written as %Now multiplying the term $W^n(\fy|\fu_{i_1})$ in (\ref{Ineq.Cond_Channel_Diff_Continuation}) by the expressions insides the bracket, we obtain
    % ---
    \begin{align}
         \label{Ineq.Cond_Channel_Diff2}
         W^n \left( \fy \, \big| \, \fu_{i_1} \right) - W^n \left( \fy \, \big| \, \fu_{i_2} \right) & \leq %W^n(\fy|\fu_{i_1}) - e^{-\sum_{t=1}^n (1+\epsilon'_n) v_{i_1,t}} \cdot (1-\epsilon'_n)^{\sum_{t=1}^n y_t} \cdot \prod_{t=1}^n \frac{v_{i_1,t}^{y_t}}{y_t!}
        % % \nonumber\\
        % % & =  W^n(\fy|\fv_{i_1}) - e^{-2\epsilon'_n\sum_{t=1}^n v_{i_1,t}} \cdot e^{-\sum_{t=1}^n (1-\epsilon'_n) v_{i_1,t}} \cdot (1-\epsilon'_n)^{\sum_{t=1}^n y_t} \cdot \prod_{t=1}^n \frac{v_{i_1,t}^{y_t}}{y_t!} 
        % \nonumber\\
        % & = W^n(\fy|\fu_{i_1}) - e^{-\epsilon'_n\sum_{t=1}^n v_{i_1,t}} \cdot (1-\epsilon'_n)^{\sum_{t=1}^n y_t} \cdot \prod_{t=1}^n \frac{e^{-v_{i_1,t}} v_{i_1,t}^{y_t}}{y_t!} 
        % \nonumber\\
        % & \stackrel{(b)}{=} W^n(\fy|\fu_{i_1}) -  e^{-\epsilon'_n\sum_{t=1}^n v_{i_1,t}} \cdot (1-\epsilon'_n)^{\sum_{t=1}^n y_t} \cdot W^n(\fy|\fu_{i_1}) %
         W^n \left( \fy \, \big| \, \fu_{i_1} \right) \cdot \left[1 - e^{-\epsilon'_n \sum_{t=1}^n v_{i_1,t}} \cdot \left( 1 - \epsilon'_n \right)^{\sum_{t=1}^n y_t} \right] 
        % \nonumber\\
        % & \leq  1 - e^{-\epsilon'_n\sum_{t=1}^n v_{i_1,t}} \cdot (1-\epsilon'_n)^{\sum_{t=1}^n y_t}  \cdot W^n(\fy|\fv_{i_1})
        \nonumber\\
        & \leq \kappa \cdot W^n \left( \fy \, \big| \, \fu_{i_1} \right) \,.
    \end{align}
    % ---
    The second inequality holds because of the following.
    Since $v_{i_1,t} \leq \lambda + P_{\,\text{max}}$ for all $t\in[\![n]\!]$, and since $\fy\in\B_{i_1}$ satisfies $\sum_{t=1}^n y_t \leq n \left( \lambda + P_{\max} + \delta \right)$, we have
    % ---
       \begin{align}
            e^{- \epsilon'_n\sum_{t=1}^n v_{i_1,t}} \cdot \left( 1 - \epsilon'_n \right)^{\sum_{t=1}^n y_t} & \geq %e^{-\epsilon'_n\sum_{t=1}^n v_{i_1,t}} \cdot (1-\epsilon'_n)^{n(\lambda + P_{\,\text{max}}+\delta)}
            %\nonumber\\
            %& \geq e^{-n\epsilon'_n(\lambda + P_{\,\text{max}})} \cdot (1-\epsilon'_n)^{nK}
            %\nonumber\\
            %& = 
            e^{-n\epsilon'_n \left( \lambda + P_{\,\text{max}} \right)} \cdot \left( 1 - \epsilon'_n \right)^{n \left( \lambda + P_{\,\text{max}} + \delta \right)}
            \nonumber\\
            & = e^{n\epsilon'_n\delta} \cdot e^{-n\epsilon'_n \left( \lambda + P_{\,\text{max}} + \delta \right)} \cdot \left( 1 - \epsilon'_n \right)^{n \left( \lambda + P_{\,\text{max}} + \delta \right)}
            %\nonumber\\
            %& \stackrel{(a)}{\geq} e^{n\epsilon'_n\delta} \cdot e^{-n\epsilon'_n (\lambda + P_{\,\text{max}}+\delta)} \cdot (1-n\epsilon'_n)^{(\lambda + P_{\,\text{max}}+\delta)}
            \nonumber\\
            & \geq e^{n\epsilon'_n \delta} \cdot f(n\epsilon'_n)
            \nonumber\\
            &\geq f(n\epsilon'_n) \,,\,
            % & \geq e^{-\epsilon'_n n(\lambda+A)} \cdot e^{\frac{nP\epsilon'_n}{\epsilon'_n - 1}}
            % \nonumber\\
            % & = e^{-\frac{\lambda+A}{n^b}} \cdot e^{\frac{nP}{1- n^{1+b}}}
            % \nonumber\\
            % & = e^{\frac{-P-\lambda-A}{n^b}}
            % \nonumber\\
            % & = e^{\frac{-\delta}{n^b}} \geq 1 - \kappa
            % \nonumber\\
            % & = e^{n\epsilon'_n \delta} - e^{-n\epsilon'_n(\lambda+A)}
            % & = \nonumber\\
            % & = e^{\frac{\delta}{n^b}} - e^{-\frac{\lambda+A}{n^b}}
            % \nonumber\\
            % \frac{\epsilon'_n}{\epsilon'_n -1} = \frac{\frac{1}{n^{1+b}}}{\frac{1}{n^{1+b}}-1} = \frac{1}{1- n^{1+b}} \sim \frac{-1}{n^{1+b}}
              \label{Ineq.Cond_Channel_Diff3}
    \end{align}
    % ---
    where
    % ---
    \begin{align}
        f(x) = e^{-cx}(1-x)^c \;,\,
    \end{align}
    % ---
    with $c = \lambda + P_{\,\text{max}}+\delta$. In (\ref{Ineq.Cond_Channel_Diff3}). The third inequality in (\ref{Ineq.Cond_Channel_Diff3}) follows by Bernoulli's inequality \cite[see Ch.~3]{Mitrinovic13}.
    Observe that by using the Taylor expansion, we have
     % ---
    \begin{align}
        f(x) = 1-2cx + \mathcal{O}(x^2) \;.\,
    \end{align}
    % ---
    Therefore, for a sufficiently small values of $x$, 
    % ---
     \begin{align}
        f(x) \geq 1-3cx \,.
    \end{align}
    % ---
    Then, the bound in  (\ref{Ineq.Cond_Channel_Diff3}) becomes
    % ---
    \begin{align}
        e^{-\epsilon'_n\sum_{t=1}^n v_{i_1,t}} \cdot \left( 1 - \epsilon'_n \right)^{\sum_{t=1}^n y_t} & \geq 1 - 3 \left( \lambda + P_{\,\text{max}} + \delta \right) n\epsilon'_n
        \nonumber\\
        & = 1 - \frac{3 \left( \lambda + P_{\,\text{max}} + \delta \right)}{n^b} 
        \nonumber\\
        & \geq 1 - \kappa \,,\,
        \label{Ineq.Cond_Channel_Diff4}
    \end{align}
    % ---
    for sufficiently large $n$, where the equality holds by (\ref{Eq.epsilonn_p}). The last bound implies (\ref{Ineq.Cond_Channel_Diff2}).
    
    Thereby, (\ref{Eq.Error_Sum_2}), (\ref{Ineq.Cond_Channel_Diff}) and (\ref{Ineq.Cond_Channel_Diff2}) together yield
    % ---
    \begin{align}
        P_{e,1}(i_1) + P_{e,2}(i_2,i_1)
        & \geq 1 - \sum_{\fy\in\B_{i_1}} \hspace{-2mm} \left[ W^n \left( \fy \, \big| \, \fu_{i_1} \right) - W^n \left( \fy \, \big| \, \fu_{i_2} \right) \right] - \kappa \nonumber \\ & = 1 - \sum_{\fy \in \B_{i_1}} \left[ \kappa \cdot W^n \left( \fy \, \big| \, \fu_{i_1} \right) \right] - \kappa
        \nonumber\\
        & \geq 1 - 2\kappa \;.\,
    \end{align}
    % ----
    Clearly, this is a contradiction since the error probabilities tend to zero as $n\rightarrow\infty$. Thus, the assumption in (\ref{Ineq.Converse_Lem_Complement}) is false. This completes the proof of Lemma~\ref{Lem.DConverseFast}.
\end{proof}
% ---
Next, we use Lemma~\ref{Lem.DConverseFast} in order to prove the upper bound on the DI capacity.
Observe that since
% ---
\begin{align}
    v_{i,t} & = u_{i,t} + \lambda
    \nonumber\\
    & > \lambda \;,\,
\end{align}
% ---
Lemma~\ref{Lem.DConverseFast} implies
% ---
\begin{align}
     \left| u_{i_1,t} - u_{i_2,t} \right| & = \left| v_{i_1,t} - v_{i_2,t} \right| 
     \nonumber\\
     & \geq \epsilon'_n v_{i_1,t}
     \nonumber\\
     & > \lambda \epsilon'_n \;.\,   
\end{align}
% ---
We deduce that the distance between every pair of codewords satisfies
% ---
\begin{align}
   \norm{\fu_{i_1} - \fu_{i_2}} & > \left| u_{i_1,t} - u_{i_2,t} \right|
   \nonumber\\
   & > \lambda \epsilon'_n \,.\,
\end{align}
% ---
Thus, we can define an arrangement of non-overlapping spheres $\S_{\fu_i}(n,\lambda \epsilon'_n)$, i.e., spheres of radius $\lambda \epsilon'_n$ that are centered at the codewords $\fu_i$.
Since the codewords all belong to a hypercube $\Q_{\f0}(n,P_{\,\text{max}})$ with edge $P_{\,\text{max}}$, it follows that the number of packed small spheres, i.e., the number of codewords $L(n,R) = 2^{(n\log n)R}$, is bounded by
% ---
\begin{align}
    2^{(n\log n)R} \leq \frac{P_{\,\text{max}}^n}{\text{Vol}(\S_{\fu_1}(n,\lambda \epsilon'_n))} \,.\,
\end{align}
% ---
Hence,
% ---
\begin{align*}
    R & \leq \frac{1}{n\log n} \cdot \log \left( \frac{P_{\,\text{max}}^n}{\text{Vol}\left(\S_{\fu_1}(n,\lambda \epsilon'_n)\right)} \right)
    \nonumber\\
    & = \frac{1}{n\log n} \left[ n\log\left( \frac{P_{\,\text{max}}}{\sqrt{\pi} r_0 } \right)+\log \left(\frac{n}{2}! \right) \right]
    \nonumber\\
    & = \frac{1}{n\log n} \left[ o(n\log n) +  \frac{1}{2}n\log n - n \log r_0 \right] \,.\,
\end{align*}
% ---
Hence, for $r_0 = \lambda \epsilon'_n = \frac{\lambda P_{\,\text{max}}}{n^{1+b}}$, we obtain
% ---
\begin{align}
    R \leq \frac{1}{n\log n} \left[ o(n\log n) +  \frac{1}{2}n\log n + (1+b)\, n \log n \right] \,,\,
\end{align}
% ---
which tends to $\frac{3}{2}$ as $n \to \infty$ and $b \to 0$.  This completes the proof of Theorem~\ref{Th.PDICapacity}.
\qed
% Since the codewords all belong to a hypercube $\Q_{\f0}(n,P_{\,\text{max}})$ with edge $P_{\,\text{max}}$, it follows that the number of packed small spheres, i.e., the number of codewords $2^{(n\log n)R}$, is bounded by
% %%%
% \begin{align}
%     2^{(n\log n)R} & \leq \frac{\text{Vol}\left[\Q_{\f0}(n,P_{\,\text{max}})\right]}{\text{Vol}(\S_{\fu_1}(n,\lambda \epsilon'_n))} = \frac{P_{\,\text{max}}^n}{\text{Vol}(\S_{\fu_1}(n,\lambda \epsilon'_n))}
% \end{align}
% %%%
% Now since $\log\text{Vol}\left(\S^n(\fu_1,r)\right) > -\frac{n}{2}\log(\frac{n}{2\pi e})-\frac{1}{2}\log n\pi - \frac{\log e}{6n} + n \log r$ \cite[See Eq.~16-18]{CHSN13}, and $ r = \lambda \epsilon'_n = \frac{\lambda P_{\,\text{max}}}{n^{1+b}}$, we obtain
% % ---
% \begin{align}
%     (n\log n)R & \leq n\log P_{\,\text{max}} +\frac{n}{2}\log(\frac{n}{2\pi e})+\frac{1}{2} \log n\pi + \frac{\log e}{6n} - n\log \lambda \epsilon'_n \,,\,
% \end{align}
% % ---
% then,
% % ---
% \begin{align}
%     R &\leq \frac{\log P_{\,\text{max}}}{\log n} + \frac{\log(\frac{n}{2\pi e})}{2\log n} + \frac{\log n\pi}{2n\log n} + \frac{\log e}{6n^2\log n} - \frac{\log \lambda \epsilon'_n}{\log n}
%     \nonumber\\
%     & \sim \frac{\log\left(n(2\pi e)^{-1}\right)}{2\log n} + \frac{\log\left(\lambda^{-1} P_{\,\text{max}}^{-1}n^{1+b}\right)}{\log n}
%     \nonumber\\
%     & \sim \frac{1}{2} + \frac{\log(n^{1+b})}{\log n}
%     \nonumber\\
%     & = \frac{1}{2} + (1+b) \cdot \frac{\log n}{\log n}
% \end{align}
% % ---
% which tends to $\frac{3}{2}$ as $n\to\infty$ and $b \to 0$.
% ------------------------------
\section{Summary and Discussion}
\label{Sec.SummaryDiscussions}

We have developed lower and upper bounds on the deterministic identification (DI) capacity of a DTPC subject to average and peak power constraints, in the scale of $L(n,R)=2^{n\log(n)R}=n^{nR}$, where $n$ is the blocklength. We have thus determined that the super-exponential scale $n^{nR}$ is the appropriate scale for the DI capacity of the Poisson channel. This scale is sharply different from the ordinary scales in the transmission and randomized-identification settings, where the code size scales exponentially and double exponentially, respectively. Different non-standard scales are also observed in other communication models, such as \emph{covert identification} \cite{ZT20,BGT13}, where the code size scales as $2^{2^{\sqrt{n}R}}$.

We have mentioned molecular communication as a motivating application for this study. Recently, there have been significant advances in molecular communication for complex nano-networks. The Internet of Things incorporates smart devices, which can be accessed and controlled via the Internet \cite{Atzori10}. The advances in nanotechnology contribute to the development of devices in the nanoscale range, referred as nanothings. The interconnection of nanothings with the Internet is known as Internet of NanoThings (IoNT) and is the basis for various future healthcare and military applications \cite{Akyildiz10,Dress15}. Nanothings are based on synthesized materials, using electronic circuits, and electromagnetic-based communication. Unfortunately, these characteristics could be harmful for some application environments, such as inside the human body. Furthermore, the concept of Internet of Bio-NanoThings (IoBNT) has been introduced in \cite{Aky15}, where nanothings are biological cells that are created using tools from synthetic biology and nanotechnology. Such biological nanothings are called bio-nanothings. Similar to artificial nanothings, bio-nanothings have control (cell nucleus), power (mitochondrion), communication (signal pathways), and sensing/actuation (flagella, pili or cilia) units. For the communication between cells, molecular communication is well suited, since the natural exchange of information between cells is already based on this paradigm. Molecular communication in cells is based on signal pathways (chains of chemical reactions) that process information that is modulated into chemical characteristics, such as molecule concentration. 

The Poisson channel is relevant for practical 6G networks in the context of MC \cite{Gohari16} and optical communications \cite{Cao13,Mceliece81}. The functionality of some 6G applications requires only a message identification. For instance, in some applications it is only required that an alert to be identified \cite{6G_PST}. In some other systems of molecular communication, a nano-device demands to determine the occurrence of a target event. In particular, in the context of MC, a nano-device might seek only to narrow down its knowledge regarding a specific task in terms of a reliable Yes/No answer. For these tasks, the identification problem is deemed as a key technology for 6G. For instance, during the targeted drug delivery \cite{Muller04,Nakano13} or cancer treatment \cite{Hobbs_ea98,Jain99,Wilhelm16}, nano-devices often should work in collaboration with each other to accomplish complex task about the treatment. This collaboration is break into simpler identification tasks such as to realize whether a certain signal or molecule was detected or not, to recognize whether a specific sort of drug is delivered or not, to verify whether a nano-machine has already replicated itself or not, to know whether the particle storage is empty or not, to discover whether the PH of blood is above a critical threshold or not, and to find out whether a target location in the vessels is identified or not and etc \cite{Nakano14}.

However, it is not clear how randomized identification (RI) codes can be incorporated into such systems. In the case of Bio-NanoThings, it is uncertain whether the natural biological processes can be controlled or reinforced by local randomness at this level. Therefore, for the design of synthetic IoNT, or for the analysis and utilization of IoBNT, identification with deterministic encoding is more appropriate.

For deterministic identification over the DTPC, the current research shows a different and unusual behavior compared to the traditional code size scales, i.e., exponential or double exponential. In the achievability, a counter-intuitive phenomenon regarding the radius of small packed hyperspheres has been observed. Surprisingly, this radius does not vanish and grows in the block-length, i.e., $\sim n^{\frac{1}{4}}$, nevertheless, the corresponding volume tends to zero more rapidly than the volume of the larger cube as $n \to \infty$. This allows us to pack super-exponentially many small hyperspheres inside the hypercube (with volume $A^n$). Observe that in the achievability part for the fading channels \cite{Salariseddigh_ITW,Salariseddigh_arXiv_ITW} the radius of the smaller spheres tends to zero in the block-length. Another interesting aspect is that even for the case of $A < 1$, which leads to a vanishing volume for the hypercube, still we are able to pack super-exponentially hyperspheres. In the converse part, we have used the argument that there is a minimum distance between the distinct codewords, in a similar manner as in the converse proof for fading channels. Nevertheless, a different upper bound is obtained.

We have considered deterministic identification for the DTPC without memory.
The DI capacity of the DTPC where $Y_t \sim \text{Pois}(\lambda + X_{t})$ subject to an average and peak power constraints $P_{\,\text{max}}$ and $P_{\,\text{ave}}$, is finite in the super exponential scale of $L(n,R)=2^{(n\log n)R}$, i.e,
% ---
\begin{align}
    \frac{1}{4}\leq\mathbb{C}_{DI}(\W,L) \leq  \frac{3}{2} \,.\;
\end{align}
% ---
The analysis for Poisson channels relies on a geometric considerations and using sphere packing. Based on fundamental properties of packing arrangements \cite{C10}, the optimal packing of non-overlapping spheres of radius $\sqrt{n\epsilon_n}$ contains an exponential number of spheres, and by decreasing the radius of the codeword spheres, the exponential rate can be made arbitrarily large. However, in the derivation of our lower bound for the $2^{n\log(n)}$-scale, we pack spheres of radius $\sqrt{n\epsilon_n}\sim n^{1/4}$, which results in $\sim 2^{\frac{1}{4}n\log(n)}$ codewords.

For the double exponential scale, or equivalently, when the rate is defined as  $R=\frac{1}{n}\log \log$ $(\#$ of messages$)$, the DI capacity is
%%%
 \begin{align}
    \mathbb{C}_{DI}(\W,L)=0  \;;\, \text{ for } L(n,R) = 2^{2^{nR}} \,,\,
 \end{align}
%%%
since the code size of DI codes scales only exponentially in block length. On the other hand, as observed by Bracher and Lapidoth \cite{BL17}, if one considers an average error criterion instead of the maximal error, then the double exponential performance of randomized-encoder codes can also be achieved using deterministic codes.

Our approach in the current research for finding positive finite bounds for the capacity was similar to our previous analysis for the Gaussian channels \cite{Salariseddigh_ITW,Salariseddigh_arXiv_ITW}. However, different analysis and bound were yielded. For instance, here in the achievability proof, we addressed packing of hyper spheres inside a hyper cube and surprisingly the radius of packed spheres scales as $\sim n^{\frac{1}{4}}$ which grows in the block-length, $n$, while the radius of small packed spheres for the Gaussian derivation tends to zero. We found that the volume for a hyper sphere of radius $\sim n^c$ vanishes for $0 < c < \frac{1}{2}$. Thereby we can pack a super-exponential number of spheres within the larger cube. The derivation in the converse proof were more involved and leaded to a larger upper bound compared to that of the Gaussian case.
% ------------------------
\section*{Acknowledgments}
Salariseddigh, Pereg, and Deppe were supported by the BMBF \text{n.~16KIS1005 (LNT, NEWCOM)}, and Boche by the BMBF \text{n.~16KIS1003K (LTI, NEWCOM)}, and the national initiative for ``Molecular Communications" (MAMOKO) n.~16KIS0914. Schober was supported by MAMOKO n.~16KIS0913.
% -----------------------------------------------------------
\appendices
% ---
\bibliographystyle{IEEEtran}
\bibliography{IEEEabrv,confs-jrnls,Lit}

\begin{thebibliography}{100}
\providecommand{\url}[1]{#1}
\csname url@samestyle\endcsname
\providecommand{\newblock}{\relax}
\providecommand{\bibinfo}[2]{#2}
\providecommand{\BIBentrySTDinterwordspacing}{\spaceskip=0pt\relax}
\providecommand{\BIBentryALTinterwordstretchfactor}{4}
\providecommand{\BIBentryALTinterwordspacing}{\spaceskip=\fontdimen2\font plus
\BIBentryALTinterwordstretchfactor\fontdimen3\font minus
  \fontdimen4\font\relax}
\providecommand{\BIBforeignlanguage}[2]{{%
\expandafter\ifx\csname l@#1\endcsname\relax
\typeout{** WARNING: IEEEtran.bst: No hyphenation pattern has been}%
\typeout{** loaded for the language `#1'. Using the pattern for}%
\typeout{** the default language instead.}%
\else
\language=\csname l@#1\endcsname
\fi
#2}}
\providecommand{\BIBdecl}{\relax}
\BIBdecl

\bibitem{Salariseddigh_arXiv_ITW}
\BIBentryALTinterwordspacing
M.~J. Salariseddigh, U.~Pereg, H.~Boche, and C.~Deppe, ``Deterministic
  identification over fading channels,'' \emph{arXiv:2010.10010}, 2020.
  [Online]. Available: \url{https://arxiv.org/pdf/2010.10010.pdf}
\BIBentrySTDinterwordspacing

\bibitem{S48}
C.~E. {Shannon}, ``A mathematical theory of communication,'' \emph{Bell Sys.
  Tech. J.}, vol.~27, no.~3, pp. 379--423, 1948.

\bibitem{AD89}
R.~{Ahlswede} and G.~{Dueck}, ``Identification via channels,'' \emph{IEEE
  Trans. Inf. Theory}, vol.~35, no.~1, pp. 15--29, 1989.

\bibitem{NMWVS12}
T.~{Nakano}, M.~J. {Moore}, F.~{Wei}, A.~V. {Vasilakos}, and J.~{Shuai},
  ``Molecular communication and networking: Opportunities and challenges,''
  \emph{IEEE Trans. Nanobiosci.}, vol.~11, no.~2, pp. 135--148, 2012.

\bibitem{FYECG16}
N.~{Farsad}, H.~B. {Yilmaz}, A.~{Eckford}, C.~{Chae}, and W.~{Guo}, ``A
  comprehensive survey of recent advancements in molecular communication,''
  \emph{IEEE Commun. Surveys Tuts.}, vol.~18, no.~3, pp. 1887--1919, 2016.

\bibitem{6G+}
W.~Haselmayr, A.~Springer, G.~Fischer, C.~Alexiou, H.~Boche, P.~A. Hoeher,
  F.~Dressler, and R.~Schober, ``Integration of molecular communications into
  future generation wireless networks,'' \emph{1st 6G Wireless Summit. IEEE,
  Levi, Finland}, 2019.

\bibitem{6G_PST}
J.~Cabrera, H.~Boche, C.~Deppe, R.~F. Schaefer, C.~Scheunert, and F.~H. Fitzek,
  ``6{G} and the {P}ost-{S}hannon {T}heory,'' in \emph{Shaping Future 6G
  Networks: Needs, Impacts and Technologies}, N.~O. Frederiksen and
  H.~Gulliksen, Eds.\hskip 1em plus 0.5em minus 0.4em\relax Hoboken, New
  Jersey, United States: Wiley-Blackwell, 2021.

\bibitem{Hobbs_ea98}
S.~K. Hobbs, W.~L. Monsky, F.~Yuan, W.~G. Roberts, L.~Griffith, V.~P.
  Torchilin, and R.~K. Jain, ``Regulation of transport pathways in tumor
  vessels: role of tumor type and microenvironment,'' \emph{Proc. Natl. Acad.
  Sci.}, vol.~95, no.~8, pp. 4607--4612, 1998.

\bibitem{Jain99}
R.~K. Jain, ``Transport of molecules, particles, and cells in solid tumors,''
  \emph{Annu. Biomed. Eng. Rev.}, vol.~1, no.~1, pp. 241--263, 1999.

\bibitem{Wilhelm16}
S.~Wilhelm, A.~J. Tavares, Q.~Dai, S.~Ohta, J.~Audet, H.~F. Dvorak, and W.~C.
  Chan, ``Analysis of nanoparticle delivery to tumours,'' \emph{Nat. Rev.
  Mater.}, vol.~1, no.~5, pp. 1--12, 2016.

\bibitem{Muller04}
R.~H. Muller and C.~M. Keck, ``Challenges and solutions for the delivery of
  biotech drugs--a review of drug nanocrystal technology and lipid
  nanoparticles,'' \emph{J. Biotech.}, vol. 113, no. 1-3, pp. 151--170, 2004.

\bibitem{Nakano13}
T.~Nakano, A.~W. Eckford, and T.~Haraguchi, \emph{Molecular
  Communication}.\hskip 1em plus 0.5em minus 0.4em\relax Cambridge University
  Press, 2013.

\bibitem{A78}
R.~Ahlswede, ``Elimination of correlation in random codes for arbitrarily
  varying channels,'' \emph{Zeitschrift f{\"u}r Wahrscheinlichkeitstheorie und
  verwandte Gebiete}, vol.~44, no.~2, pp. 159--175, 1978.

\bibitem{feedback}
R.~{Ahlswede} and G.~{Dueck}, ``Identification in the presence of feedback-a
  discovery of new capacity formulas,'' \emph{IEEE Trans. Inf. Theory},
  vol.~35, no.~1, pp. 30--36, Jan 1989.

\bibitem{correlation}
H.~{Boche}, R.~F. {Schaefer}, and H.~{Vincent Poor}, ``On the computability of
  the secret key capacity under rate constraints,'' in \emph{IEEE Int. Conf.
  Acoust. Speech Sig. Proc. (ICASSP)}, May 2019, pp. 2427--2431.

\bibitem{BV00}
M.~V. {Burnashev}, ``On identification capacity of infinite alphabets or
  continuous-time channels,'' \emph{IEEE Trans. Inf. Theory}, vol.~46, no.~7,
  pp. 2407--2414, 2000.

\bibitem{BD18_2}
H.~{Boche} and C.~{Deppe}, ``Secure identification for wiretap channels;
  robustness, super-additivity and continuity,'' \emph{IEEE Trans. Inf. Foren.
  Secur.}, vol.~13, no.~7, pp. 1641--1655, 2018.

\bibitem{W04}
A.~Winter, ``Quantum and classical message identification via quantum
  channels,'' \emph{arXiv preprint quant-ph/0401060}, 2004.

\bibitem{BL17}
A.~Bracher and A.~Lapidoth, ``Identification via the broadcast channel,''
  \emph{IEEE Trans. Inf. Theory}, vol.~63, no.~6, pp. 3480--3501, 2017.

\bibitem{S56}
C.~{Shannon}, ``The zero error capacity of a noisy channel,'' \emph{IRE Trans.
  Inf. Theory}, vol.~2, no.~3, pp. 8--19, 1956.

\bibitem{SCR20-2}
S.~Derebeyo{\u{g}}lu, C.~Deppe, and R.~Ferrara, ``Performance analysis of
  identification codes,'' \emph{Entropy}, vol.~22, no.~10, p. 1067, 2020.

\bibitem{VK93}
S.~{Verd\'u} and V.~K. {Wei}, ``Explicit construction of optimal
  constant-weight codes for identification via channels,'' \emph{IEEE Trans.
  Inf. Theory}, vol.~39, no.~1, pp. 30--36, 1993.

\bibitem{KT99}
K.~{Kurosawa} and T.~{Yoshida}, ``Strongly universal hashing and identification
  codes via channels,'' \emph{IEEE Trans. Inf. Theory}, vol.~45, no.~6, pp.
  2091--2095, 1999.

\bibitem{Bringer09}
J.~Bringer, H.~Chabanne, G.~Cohen, and B.~Kindarji, ``Private interrogation of
  devices via identification codes,'' in \emph{Int. Conf. Cryptol. in
  India}.\hskip 1em plus 0.5em minus 0.4em\relax Springer, 2009, pp. 272--289.

\bibitem{Bringer10}
------, ``Identification codes in cryptographic protocols,'' in \emph{IEEE Inf.
  Theory Workshop}, 2010, pp. 1--5.

\bibitem{MasterThesis}
W.~Labidi, ``Secure {I}dentification for {G}aussian {C}hannels,'' Master's
  thesis, LNT, Technical University of Munich (TUM), June 2019.

\bibitem{LDB20}
W.~{Labidi}, C.~{Deppe}, and H.~{Boche}, ``Secure identification for {G}aussian
  channels,'' in \emph{IEEE Int. Conf. Acoust. Speech Sig. Proc. (ICASSP)},
  2020, pp. 2872--2876.

\bibitem{Labidi2021}
W.~Labidi, H.~Boche, C.~Deppe, and M.~Wiese, ``Identification over the
  {G}aussian channel in the presence of feedback,'' \emph{IEEE Int. Symp. Inf.
  Theory (ISIT 2021). Preprint available on arXiv:2102.01198}, 2021.

\bibitem{Ezzine2021}
R.~Ezzine, W.~Labidi, H.~Boche, and C.~Deppe, ``Common randomness generation
  and identification over {G}aussian channels,'' in \emph{Proc. IEEE Global
  Comm. Conf.}\hskip 1em plus 0.5em minus 0.4em\relax IEEE, 2020, pp. 1--6.

\bibitem{Salariseddigh_ITW}
M.~J. Salariseddigh, U.~Pereg, H.~Boche, and C.~Deppe, ``Deterministic
  identification over fading channels,'' in \emph{2020-IEEE Int'l Inf. Theory
  Workshop (ITW)}.\hskip 1em plus 0.5em minus 0.4em\relax IEEE, 2021, pp. 1--5.

\bibitem{AN99}
R.~{Ahlswede} and {Ning Cai}, ``Identification without randomization,''
  \emph{IEEE Trans. Inf. Theory}, vol.~45, no.~7, pp. 2636--2642, 1999.

\bibitem{Salariseddigh_arXiv_ICC}
\BIBentryALTinterwordspacing
M.~J. Salariseddigh, U.~Pereg, H.~Boche, and C.~Deppe, ``Deterministic
  identification over channels with power constraints,'' \emph{Submitted to
  IEEE Trans. Inf. Theory, arXiv:2010.04239}, 2020. [Online]. Available:
  \url{https://arxiv.org/pdf/2010.04239.pdf}
\BIBentrySTDinterwordspacing

\bibitem{J85}
J.~J\'aJ\'a, ``Identification is easier than decoding,'' in \emph{Ann. Symp.
  Found. Comp. Scien. (SFCS)}, 1985, pp. 43--50.

\bibitem{Bur00}
M.~V. Burnashev, ``On the method of types and approximation of output measures
  for channels with finite alphabets,'' \emph{Prob. Inf. Trans.}, vol.~36,
  no.~3, pp. 195--212, 2000.

\bibitem{PP09}
R.~L. Bocchino~Jr, V.~S. Adve, S.~V. Adve, and M.~Snir, ``Parallel programming
  must be deterministic by default,'' in \emph{Proc. 1st USENIC Conf.: Hot
  Topics in Parallelism}, 2009, pp. 4--4.

\bibitem{A09}
E.~{Ar{\i}kan}, ``Channel polarization: A method for constructing
  capacity-achieving codes for symmetric binary-input memoryless channels,''
  \emph{IEEE Trans. Inf. Theory}, vol.~55, no.~7, pp. 3051--3073, 2009.

\bibitem{Salariseddigh_ICC}
M.~J. Salariseddigh, U.~Pereg, H.~Boche, and C.~Deppe, ``Deterministic
  identification over channels with power constraints,'' in \emph{2021-IEEE
  Int'l Conf. Commun. (ICC)}.\hskip 1em plus 0.5em minus 0.4em\relax IEEE,
  2021, pp. 1--6.

\bibitem{Fillmore69}
G.~Fillmore and G.~Lachs, ``Information rates for photocount detection
  systems,'' \emph{IEEE Trans. Inf. Theory}, vol.~15, no.~4, pp. 465--468,
  1969.

\bibitem{Gagliardi76}
R.~M. Gagliardi and S.~Karp, ``Optical communications,'' \emph{New York}, 1976.

\bibitem{Shamai90}
S.~S. Shamai, ``Capacity of a pulse amplitude modulated direct detection photon
  channel,'' \emph{IEEE Proc. I (Commun., Speech and Vision)}, vol. 137, no.~6,
  pp. 424--430, 1990.

\bibitem{Shapiro09}
J.~H. Shapiro, ``The quantum theory of optical communications,'' \emph{IEEE J.
  Sel. Topics Quantum Electron.}, vol.~15, no.~6, pp. 1547--1569, 2009.

\bibitem{Wyner88_I}
A.~D. Wyner, ``Capacity and error exponent for the direct detection photon
  channel. i,'' \emph{IEEE Trans. Inf. Theory}, vol.~34, no.~6, pp. 1449--1461,
  1988.

\bibitem{Wyner88_II}
------, ``Capacity and error exponent for the direct detection photon channel.
  ii,'' \emph{IEEE Trans. Inf. Theory}, vol.~34, no.~6, pp. 1462--1471, 1988.

\bibitem{Massey81}
J.~Massey, ``Capacity, cutoff rate, and coding for a direct-detection optical
  channel,'' \emph{IEEE Trans. Commun.}, vol.~29, no.~11, pp. 1615--1621, 1981.

\bibitem{Verdu99}
S.~Verd{\'u}, ``{P}oisson communication theory,'' presented at the
  International Technion Communication Day in Honor of Israel Bar-David, Haifa,
  Israel, Mar. 25, 1999, Invited Talk, 1999.

\bibitem{Bar73}
I.~Bar-David, ``Information in the time of arrival of a photon packet: Capacity
  of ppm channels,'' \emph{JOSA}, vol.~63, no.~2, pp. 166--170, 1973.

\bibitem{Pierce81}
J.~Pierce, E.~Posner, and E.~Rodemich, ``The capacity of the photon counting
  channel,'' \emph{IEEE Trans. Inf. Theory}, vol.~27, no.~1, pp. 61--77, 1981.

\bibitem{Kabanov78}
Y.~M. Kabanov, ``The capacity of a channel of the {P}oisson type,''
  \emph{Theory of Probability \& Its Applications}, vol.~23, no.~1, pp.
  143--147, 1978.

\bibitem{Davis80}
M.~Davis, ``Capacity and cutoff rate for {P}oisson-type channels,'' \emph{IEEE
  Trans. Inf. Theory}, vol.~26, no.~6, pp. 710--715, 1980.

\bibitem{Wyner84}
A.~Wyner and H.~Landau, ``Optimum waveform signal sets with amplitude and
  energy constraints,'' \emph{IEEE Trans. Inf. Theory}, vol.~30, no.~4, pp.
  615--622, 1984.

\bibitem{Butman82}
S.~Butman, J.~Katz, and J.~Lesh, ``Bandwidth limitations on noiseless optical
  channel capacity,'' \emph{IEEE Trans. Inf. Theory}, vol.~30, no.~5, pp.
  1262--1264, 1982.

\bibitem{Lesh83}
J.~Lesh, ``Capacity limit of the noiseless, energy-efficient optical ppm
  channel,'' \emph{IEEE Trans. Commun.}, vol.~31, no.~4, pp. 546--548, 1983.

\bibitem{Lipes80}
R.~Lipes, ``Pulse-position-modulation coding as near-optimum utilization of
  photon counting channel with bandwidth and power constraints,'' \emph{The
  Deep Space Network Progress Report 42-56, January and February 1980}, pp.
  108--113, 1980.

\bibitem{Shamai91}
S.~Shamai, ``On the capacity of a direct-detection photon channel with
  intertransition-constrained binary input,'' \emph{IEEE Trans. Inf. Theory},
  vol.~37, no.~6, pp. 1540--1550, 1991.

\bibitem{Shamai93}
S.~Shamai and A.~Lapidoth, ``Bounds on the capacity of a spectrally constrained
  {P}oisson channel,'' \emph{IEEE Trans. Inf. Theory}, vol.~39, no.~1, pp.
  19--29, 1993.

\bibitem{Lapidoth93}
A.~Lapidoth, ``On the reliability function of the ideal {P}oisson channel with
  noiseless feedback,'' \emph{IEEE Trans. Inf. Theory}, vol.~39, no.~2, pp.
  491--503, 1993.

\bibitem{Lapidoth98}
A.~Lapidoth and S.~Shamai, ``The {P}oisson multiple-access channel,''
  \emph{IEEE Trans. Inf. Theory}, vol.~44, no.~2, pp. 488--501, 1998.

\bibitem{Lapidoth03_4}
A.~Lapidoth, I.~E. Telatar, and R.~Urbanke, ``On wide-band broadcast
  channels,'' \emph{IEEE Trans. Inf. Theory}, vol.~49, no.~12, pp. 3250--3258,
  2003.

\bibitem{Bross09}
S.~Bross, A.~Lapidoth, and L.~Wang, ``The {P}oisson channel with side
  information,'' in \emph{47th Annu. Allerton Conf. Commun. Control
  Comput.}\hskip 1em plus 0.5em minus 0.4em\relax IEEE, 2009, pp. 574--578.

\bibitem{Frey91}
M.~R. Frey, ``Information capacity of the {P}oisson channel,'' \emph{IEEE
  Trans. Inf. Theory}, vol.~37, no.~2, pp. 244--256, 1991.

\bibitem{Frey92}
------, ``Capacity of the l/sub p/norm-constrained {P}oisson channel,''
  \emph{IEEE Trans. Inf. Theory}, vol.~38, no.~2, pp. 445--450, 1992.

\bibitem{Brady90}
D.~Brady and S.~Verd{\'u}, ``The asymptotic capacity of the direct detection
  photon channel with a bandwidth constraint,'' in \emph{Proc. 28th Allerton
  Conf. Commun. Control Comput.}, 1990, pp. 691--700.

\bibitem{Lapidoth03_1}
A.~Lapidoth and S.~M. Moser, ``Bounds on the capacity of the discrete-time
  {P}oisson channel,'' in \emph{Proc. Allerton Conf. Comm., Control,
  Computing}, vol.~41, no.~1.\hskip 1em plus 0.5em minus 0.4em\relax The
  University; 1998, 2003, pp. 201--210.

\bibitem{Lapidoth03_2}
------, ``The asymptotic capacity of the discrete-time {P}oisson channel,''
  \emph{Proc. Win. Sch. Coding and Inf. Theory, Monte Verita, Ascona,
  Switzerland}, 2003.

\bibitem{Lapidoth08_1}
A.~Lapidoth, L.~Wang, J.~H. Shapiro, and V.~Venkatesan, ``The {P}oisson channel
  at low input powers,'' in \emph{IEEE 25th Conv. Elect. and Electron.
  Engineers in Israel}.\hskip 1em plus 0.5em minus 0.4em\relax IEEE, 2008, pp.
  654--658.

\bibitem{Lapidoth11}
A.~Lapidoth, J.~H. Shapiro, V.~Venkatesan, and L.~Wang, ``The discrete-time
  {P}oisson channel at low input powers,'' \emph{IEEE Trans. Inf. Theory},
  vol.~57, no.~6, pp. 3260--3272, 2011.

\bibitem{Ahmadypour20}
N.~Ahmadypour and A.~Gohari, ``Transmission of a bit over a discrete {P}oisson
  channel with memory,'' \emph{arXiv preprint arXiv:2011.05931}, 2020.

\bibitem{Etemadi19}
A.~Etemadi, H.~Arjmandi, P.~Azmi, and N.~Mokari, ``Capacity bounds for
  diffusive molecular communication over discrete-time compound {P}oisson
  channels,'' \emph{IEEE Comm. Lett.}, vol.~23, no.~5, pp. 793--796, 2019.

\bibitem{Chakraborty07}
K.~Chakraborty and P.~Narayan, ``The {P}oisson fading channel,'' \emph{IEEE
  Trans. Inf. Theory}, vol.~53, no.~7, pp. 2349--2364, 2007.

\bibitem{Wang21}
L.~Wang, ``Covert communication over the {P}oisson channel,'' \emph{IEEE J.
  Sel. Areas Inf. Theory}, vol.~2, no.~1, pp. 23--31, 2021.

\bibitem{Gohari16}
A.~Gohari, M.~Mirmohseni, and M.~Nasiri-Kenari, ``Information theory of
  molecular communication: Directions and challenges,'' \emph{IEEE Trans. Mol.
  Biol. Multi-Scale Commun.}, vol.~2, no.~2, pp. 120--142, 2016.

\bibitem{Cao13}
J.~Cao, ``Discrete-time {P}oisson channel: capacity and signalling design,''
  Ph.D. dissertation, McMaster University, 2013.

\bibitem{Mceliece81}
R.~McEliece, ``Practical codes for photon communication,'' \emph{IEEE Trans.
  Inf. Theory}, vol.~27, no.~4, pp. 393--398, 1981.

\bibitem{Nakano14}
T.~Nakano, T.~Suda, Y.~Okaie, M.~J. Moore, and A.~V. Vasilakos, ``Molecular
  communication among biological nanomachines: A layered architecture and
  research issues,'' \emph{IEEE Trans. Nanobiosci.}, vol.~13, no.~3, pp.
  169--197, 2014.

\bibitem{Schulten00}
K.~Schulten and I.~Kosztin, ``Lectures in theoretical biophysics,''
  \emph{University of Illinois}, vol. 117, 2000.

\bibitem{Redner01}
S.~Redner, \emph{A guide to first-passage processes}.\hskip 1em plus 0.5em
  minus 0.4em\relax Cambridge University Press, 2001.

\bibitem{Wang14}
L.~Wang and G.~W. Wornell, ``A refined analysis of the {P}oisson channel in the
  high-photon-efficiency regime,'' \emph{IEEE Trans. Inf. Theory}, vol.~60,
  no.~7, pp. 4299--4311, 2014.

\bibitem{Aminian15}
G.~Aminian, H.~Arjmandi, A.~Gohari, M.~Nasiri-Kenari, and U.~Mitra, ``Capacity
  of diffusion-based molecular communication networks over lti-{P}oisson
  channels,'' \emph{IEEE Trans. Mol. Biol. Multi-Scale Commun.}, vol.~1, no.~2,
  pp. 188--201, 2015.

\bibitem{Martinez07}
A.~Martinez, ``Spectral efficiency of optical direct detection,'' \emph{JOSA
  B}, vol.~24, no.~4, pp. 739--749, 2007.

\bibitem{Cao10}
J.~Cao, S.~Hranilovic, and J.~Chen, ``Lower bounds on the capacity of
  discrete-time {P}oisson channels with dark current,'' in \emph{25th Biennial
  Symp. Commun.}\hskip 1em plus 0.5em minus 0.4em\relax IEEE, 2010, pp.
  357--360.

\bibitem{Yu14}
Y.~Yu, Z.~Zhang, L.~Wu, and J.~Dang, ``Lower bounds on the capacity for
  {P}oisson optical channel,'' in \emph{Sixth Int'l Conf. Wireless Commun. and
  Sig. Proc. (WCSP)}.\hskip 1em plus 0.5em minus 0.4em\relax IEEE, 2014, pp.
  1--5.

\bibitem{Cheraghchi19}
M.~Cheraghchi and J.~Ribeiro, ``Improved upper bounds and structural results on
  the capacity of the discrete-time {P}oisson channel,'' \emph{IEEE Trans. Inf.
  Theory}, vol.~65, no.~7, pp. 4052--4068, 2019.

\bibitem{Cheraghchi20}
------, ``Non-asymptotic capacity upper bounds for the discrete-time poisson
  channel with positive dark current,'' \emph{arXiv preprint arXiv:2010.14858},
  2020.

\bibitem{Smith71}
J.~G. Smith, ``The information capacity of amplitude-and variance-constrained
  sclar {G}aussian channels,'' \emph{Information and control}, vol.~18, no.~3,
  pp. 203--219, 1971.

\bibitem{Abou01}
I.~C. Abou-Faycal, M.~D. Trott, and S.~Shamai, ``The capacity of discrete-time
  memoryless rayleigh-fading channels,'' \emph{IEEE Trans. Inf. Theory},
  vol.~47, no.~4, pp. 1290--1301, 2001.

\bibitem{Elmoslimany17}
A.~ElMoslimany and T.~M. Duman, ``On the discreteness of capacity-achieving
  distributions for fading and signal-dependent noise channels with
  amplitude-limited inputs,'' \emph{IEEE Trans. Inf. Theory}, vol.~64, no.~2,
  pp. 1163--1177, 2017.

\bibitem{Fahs17}
J.~Fahs and I.~Abou-Faycal, ``On properties of the support of
  capacity-achieving distributions for additive noise channel models with input
  cost constraints,'' \emph{IEEE Trans. Inf. Theory}, vol.~64, no.~2, pp.
  1178--1198, 2017.

\bibitem{Dytso18}
A.~Dytso, M.~Goldenbaum, H.~V. Poor, and S.~S. Shitz, ``When are discrete
  channel inputs optimal?—optimization techniques and some new results,'' in
  \emph{52nd Annu. Conf. Inf. Scien. Sys. (CISS)}.\hskip 1em plus 0.5em minus
  0.4em\relax IEEE, 2018, pp. 1--6.

\bibitem{Dytso21}
A.~Dytso, L.~Barletta, and S.~Shamai, ``Bounds on the number of mass points of
  the capacity achieving distribution of the amplitude constraint {P}oisson
  noise channel,'' \emph{arXiv preprint arXiv:2104.14431}, 2021.

\bibitem{Cao13_1}
J.~Cao, S.~Hranilovic, and J.~Chen, ``Capacity-achieving distributions for the
  discrete-time {P}oisson channel—part i: General properties and numerical
  techniques,'' \emph{IEEE Trans. Commun.}, vol.~62, no.~1, pp. 194--202, 2013.

\bibitem{Cao13_2}
------, ``Capacity-achieving distributions for the discrete-time {P}oisson
  channel—part ii: Binary inputs,'' \emph{IEEE Trans. Commun.}, vol.~62,
  no.~1, pp. 203--213, 2013.

\bibitem{Lapidoth08_2}
A.~Lapidoth and S.~M. Moser, ``On the capacity of the discrete-time {P}oisson
  channel,'' \emph{IEEE Trans. Inf. Theory}, vol.~55, no.~1, pp. 303--322,
  2008.

\bibitem{Brady90_PhD_Diss}
\BIBentryALTinterwordspacing
D.~Brady, ``The analysis of optical, direct detection communication systems
  with point process observations,'' Ph.D. dissertation, Princeton University,
  1990. [Online]. Available:
  \url{https://search.proquest.com/docview/303849140}
\BIBentrySTDinterwordspacing

\bibitem{CHSN13}
J.~H. Conway and N.~J.~A. Sloane, \emph{Sphere packings, lattices and
  groups}.\hskip 1em plus 0.5em minus 0.4em\relax Springer Science \& Business
  Media, 2013, vol. 290.

\bibitem{C10}
H.~Cohn, ``Order and disorder in energy minimization,'' in \emph{Int. Congr.
  Mathematicians. Proc. (ICM) (In 4 Volumes) Vol. I: Plenary Lectures and
  Ceremonies Vols. II--IV: Invited Lectures}.\hskip 1em plus 0.5em minus
  0.4em\relax World Scientific, 2010, pp. 2416--2443.

\bibitem{Mitrinovic13}
D.~S. Mitrinovic, J.~Pecaric, and A.~M. Fink, \emph{Classical and new
  inequalities in analysis}.\hskip 1em plus 0.5em minus 0.4em\relax Springer
  Sci. \& Bus. Media, 2013, vol.~61.

\bibitem{ZT20}
Q.~Zhang and V.~Y. Tan, ``Covert identification over binary-input memoryless
  channels,'' \emph{arXiv preprint arXiv:2007.13333}, 2020.

\bibitem{BGT13}
B.~A. Bash, D.~Goeckel, and D.~Towsley, ``Limits of reliable communication with
  low probability of detection on awgn channels,'' \emph{IEEE J. Sel. Areas
  Commun.}, vol.~31, no.~9, pp. 1921--1930, 2013.

\bibitem{Atzori10}
\BIBentryALTinterwordspacing
L.~Atzori, A.~Iera, and G.~Morabito, ``The internet of things: A survey,''
  \emph{Computer Networks}, vol.~54, no.~15, pp. 2787--2805, 2010. [Online].
  Available:
  \url{https://www.sciencedirect.com/science/article/pii/S1389128610001568}
\BIBentrySTDinterwordspacing

\bibitem{Akyildiz10}
I.~F. Akyildiz and J.~M. Jornet, ``The internet of nano-things,'' \emph{IEEE
  Wireless Commun.}, vol.~17, no.~6, pp. 58--63, 2010.

\bibitem{Dress15}
F.~Dressler and S.~Fischer, ``Connecting in-body nano communication with body
  area networks: Challenges and opportunities of the internet of nano things,''
  \emph{Nano Commun. Networks}, vol.~6, pp. 29--38, 2015.

\bibitem{Aky15}
I.~Akyildiz, M.~Pierobon, S.~Balasubramaniam, and Y.~Koucheryavy, ``The
  internet of bio-nano things,'' \emph{IEEE Commun. Mag.}, vol.~53, pp. 32--40,
  2015.

\end{thebibliography}
% ---
\end{document}